\ifdefined\CPCReviewLayout
\documentclass[preprint,12pt,3p]{elsarticle}
\else
\documentclass[final,5p,times,twocolumn]{elsarticle}
\fi

\usepackage{mathtools}
\usepackage{placeins}
\usepackage{algpseudocode,algorithmicx,algorithm}
\usepackage{mathrsfs}
\usepackage{latexsym,bm,bbm}
\usepackage{amsmath,amsfonts,amssymb,amsthm,rotating}
\usepackage{extarrows}
\usepackage{graphicx,subfigure,epstopdf,float}
\graphicspath{{images/}{./}}
\usepackage{enumerate,cases,multirow}
\usepackage{makecell}
\usepackage{caption}
\usepackage{longtable,colortbl,arydshln,threeparttable}
\definecolor{mygray}{gray}{.9}
\usepackage{listings}
\usepackage{makeidx}
\usepackage{array,booktabs}
\usepackage{tikz}
\usepackage{xcolor}
\usepackage[bookmarksnumbered,colorlinks=true,citecolor=blue,linkcolor=blue,urlcolor=blue,hyperindex,linktocpage=true]{hyperref}
\newcommand{\ket}[1]{| #1 \rangle}
\newcommand{\bra}[1]{\langle #1 |}

\newcommand{\bb}{\boldsymbol}

\def \d {\mathrm{d}}
\def \e {\mathrm{e}}
\def \i {\ensuremath{\mathrm{i}}}

\newcounter{parentalgorithm}

\newtheorem{theorem}{Theorem}[section]
\newtheorem{lemma}{Lemma}[section]
\newtheorem{proposition}{Proposition}[section]

\theoremstyle{remark}
\newtheorem{remark}{\bf Remark}[section]

\numberwithin{equation}{section}

\journal{Computer Physics Communications}

\begin{document}

\begin{frontmatter}

\title{Equilibrium-energy-scaled quantum lattice Boltzmann method for variable-velocity linear advection--diffusion equations}

\author[aff1]{Yuan He}
\ead{heyuan@smail.xtu.edu.cn}

\author[aff1,aff2,aff3]{Yuan Yu\corref{cor1}}
\ead{yuyuan@xtu.edu.cn}
\cortext[cor1]{Corresponding author.}

\author[aff1,aff2,aff3]{Yue Yu}
\ead{terenceyuyue@xtu.edu.cn}

\author[aff1]{Yuting Zhou}
\ead{202521511280@smail.xtu.edu.cn}

\affiliation[aff1]{organization={School of Mathematics and Computational Science, Xiangtan University},
            city={Xiangtan},
            postcode={411105},
            state={Hunan},
            country={China}}

\affiliation[aff2]{organization={Hunan Research Center of the Basic Discipline Fundamental Algorithmic Theory and Novel Computational Methods},
            city={Xiangtan},
            postcode={411105},
            state={Hunan},
            country={China}}

\affiliation[aff3]{organization={National Center for Applied Mathematics in Hunan},
            city={Xiangtan},
            postcode={411105},
            state={Hunan},
            country={China}}

\begin{abstract}
We develop a quantum lattice Boltzmann method for linear advection--diffusion equations with prescribed space- and time-dependent velocities. Direct block-encoding can accumulate normalization factors that suppress success probability exponentially, even when the classical propagator remains bounded. We control this loss by scaling each distribution component by the inverse square root of its local equilibrium weight. For uniformly positive weights and dimensionless relaxation times $\tau^*\geq 1/2$, this equilibrium energy scaling gives the Bhatnagar--Gross--Krook collision operator unit spectral norm while preserving the discrete dynamics. Streaming normalization depends on weight ratios at the departure and arrival positions and times of each step; these ratios equal one for uniform steady advection. For periodic streaming, cumulative normalization stays bounded under mesh refinement at fixed physical time if logarithmic derivatives of the weights along streaming trajectories are bounded uniformly in space, time, and mesh resolution. A single-step block-encoding with two flag qubits supports coherent propagation of the full kinetic distribution without intermediate measurements. We compare oracle query bounds for time marching with those for a history-state linear system built from the same scaled operator, under specified assumptions on initial-state preparation and coefficient access. Tests in one and two dimensions and noiseless circuit simulations verify the construction. In a controlled comparison after 50 steps, scaling increases the probability of preparing the quantum state encoding the terminal distribution by approximately $450\times$. Classical refinement diagnostics up to $N_t=16{,}384$ steps confirm bounded cumulative normalization. The scaling controls additional normalization losses while retaining physical diffusive attenuation.
\end{abstract}


\end{frontmatter}

\section{Introduction}
\begingroup
\clubpenalty=10000
\widowpenalty=10000

Passive scalar transport describes heat transfer, tracer dispersion, and species transport in a prescribed flow. The lattice Boltzmann method (LBM) represents this transport through local relaxation and nearest-neighbor streaming of kinetic distributions \cite{kruger2017lattice,ginzburg2005equilibrium,cartalade2019multiple,sterling1996stability,chai2013lattice}. Fine meshes increase both the number of unknowns and the cost of repeated updates \cite{chan1984stability,dehghan2004numerical}. Quantum lattice Boltzmann methods (QLBMs) encode spatial indices in qubit registers whose width grows logarithmically with the number of grid points \cite{todorova2020quantum,budinski2021quantum,schalkers2024importance,georgescu2025qlbm}. Using this compact representation for transport requires an evolution algorithm and a specified output task, such as preparing a normalized terminal state or estimating a macroscopic functional \cite{jin2022quantum,auyeung2024quantum}. %[STAKES:L01]

The state needed for that evolution depends on the relaxation model. With Bhatnagar--Gross--Krook (BGK) relaxation, full relaxation at dimensionless relaxation time $\tau^*=1$ projects the distribution onto local equilibrium, which is determined by the scalar field and prescribed velocity. For $\tau^*\neq1$, non-equilibrium components survive collision and carry information between time steps \cite{matyka2021memory}. Two kinetic distributions with the same scalar field can therefore generate different subsequent scalar fields. Propagating the full distribution retains this information when extending a quantum algorithm beyond full relaxation. %[L01]

Encoding the full kinetic distribution requires control of the normalization introduced at each time step. A block-encoding represents each nonunitary update divided by a factor at least as large as its Euclidean spectral norm. Classical mass conservation and macroscopic consistency alone do not ensure that this norm is at most one \cite{huang2014modified,rao2015numerical,bellotti2024modified}. Products of stepwise normalization factors can then cause exponentially small success probabilities even when the ordered classical propagator remains bounded. Controlling this additional loss requires a suitable representation of the kinetic update. %[L02,L03]

General quantum frameworks provide operator constructions for nonunitary propagation through linear combinations of unitaries, block-encoding, and quantum singular value transformation \cite{Low2019Interaction,Gilyen2019QSVD,Chakraborty2019blockEncode}. Hamiltonian simulation and linear systems encoding the evolution history provide complementary routes to time evolution \cite{ACL2023LCH2,berry2014high,Costa2021QLSA,jin2022quantum}. For lattice Boltzmann models, unitary collision embeddings and protocols for successive steps have made the representation and propagation of kinetic information explicit at the circuit level \cite{wawrzyniak2024unitary,kocherla2024fully,kumar2025quantum,georgescu2025fully}. %[BG:T01-r10] %[BG:T01-r11]

Recent transport algorithms have connected these general tools to explicit circuits and normalization control. Over et al.\ \cite{over2025direct} construct a direct block-encoding for a finite-difference advection--diffusion update that removes an additional normalization factor, so that success probability is determined, up to the encoding error, by the ratio of final and initial squared solution norms. Their construction accommodates spatially varying velocities and combines an approximate Hamiltonian embedding with shift operators for forward Euler time stepping and central spatial differences. Explicit circuits for non-Hermitian advection--diffusion operators \cite{novikau2025quantum} and an efficient implementation of linear combinations of Hamiltonian simulations \cite{novikau2026explicit} further connect nonunitary propagation to verifiable circuit constructions. In the latter, a coordinate transformation simplifies block-encoding and combines the Hamiltonian evolutions within a single quantum signal processing circuit. Variable transformations also produce stable semidiscrete systems for Schr\"odingerization of time-fractional heat equations \cite{jin2026schrodingerization}. These developments show how the choice of representation can improve both operator structure and its quantum implementation. %[BG:T01-r01] %[BG:T01-r07] %[BG:T01-r08] %[BG:T01-r09]

Several scalar-field QLBM constructions exploit full relaxation to simplify collision and support repeated updates. Wawrzyniak et al.\ \cite{wawrzyniak2025linearized} use square-root density encoding and measurement-controlled streaming to accommodate spatially varying velocities without collision postselection or full-state readout and reinitialization at every step. Nagel and L\"owe \cite{nagel2026multiple} retain residual information in a time register, enabling coherent multi-step evolution without intermediate measurements or reinitialization. Tiwari et al.\ \cite{tiwari2025algorithmic} remove an additional exponential normalization factor so that the success probability of their multi-step protocol reduces to a ratio of final and initial squared scalar norms. Their work also develops state-preparation, gate, and readout improvements and demonstrates transport on quantum hardware. Xu et al.\ \cite{xu2025improved} use unitary collision preparation to remove the collision auxiliary qubit and avoid intermediate full-state tomography when advancing to a prescribed final step. %[BG:T01-r02] %[BG:T01-r03] %[BG:T01-r04] %[BG:T01-r05]

Weighted kinetic representations address normalization through the geometry of the collision operator. Yao and Succi~\cite{yao2026representation} combine fixed equilibrium-weight scaling $f_i/\sqrt{w_i}$ with a Hermite basis to obtain a linear relaxation block of unit norm in a Carleman formulation. They further scale the quadratic components to control coupling and derive polynomial encoding and readout overhead in the number of steps and lattice size. Their analysis therefore connects improved single-step success to the feasibility of multi-step Carleman evolution, with physical dissipation and the second-order truncation determining the remaining cost and accuracy. Table~\ref{tab:relatedConstructions} summarizes the capabilities and assumptions of selected QLBM constructions. %[BG:T01-r06]

\begin{table*}[!htb]
\centering\small
\caption{Capabilities and assumptions of selected QLBM constructions. State representation and temporal propagation are listed alongside the normalization mechanism.}
\label{tab:relatedConstructions}
\setlength{\tabcolsep}{5pt}
\begin{tabular}{>{\raggedright\arraybackslash}p{2.5cm}>{\raggedright\arraybackslash}p{5.0cm}>{\raggedright\arraybackslash}p{6.5cm}}
\toprule
Construction & State and temporal propagation & Normalization mechanism and scope \\
\midrule
Tiwari et al.\ \cite{tiwari2025algorithmic} & Scalar amplitudes; full relaxation; direction measurements after each step. & Square-root coefficient preparation removes an additional exponential factor, leaving a ratio of final and initial squared scalar norms. Nonuniform coefficients satisfy a partition constraint after spatial shifts. \\
Xu et al.\ \cite{xu2025improved} & Scalar amplitudes; full relaxation; subsequent loops conditioned on successful measurement. & Unitary collision preparation removes its auxiliary flag. Conditional successive updates avoid intermediate full-state tomography; macroscopic summation uses postselection. \\
Yao and Succi~\cite{yao2026representation} & Carleman state truncated at second order; fixed weights $w_i$; multiple-relaxation-time or BGK collision; postselection on the ancillas at each step. & Weighted Hermite coordinates give the linear relaxation block unit norm. Scaling the quadratic components yields polynomial encoding and readout overhead across multiple steps; physical dissipation and truncation affect cost and accuracy. \\
Present analysis & Full distribution of the linear BGK model; dynamic weights $a_i(t,\bb x)$; coherent composition and terminal output. & Positive weights give the scaled collision operator unit norm for $\tau^*\geq1/2$. Ratios of weights between streaming departure and arrival points control cumulative normalization under uniform bounds on their logarithmic derivatives. \\
\bottomrule
\end{tabular}
\end{table*}

These results establish several ways to propagate transport states while reducing losses introduced by their quantum encoding. The present analysis concerns the full distribution of the linear BGK model at general admissible relaxation times under prescribed space- and time-dependent advection. In this setting, local equilibrium weights change both during streaming and between time steps. Keeping cumulative normalization bounded under mesh refinement requires relating the collision operator in a weighted norm to these changes along streaming trajectories. %[L01,L02,L03]

To meet these requirements, we scale each distribution component by the inverse square root of its prescribed local equilibrium weight. This equilibrium energy scaling defines a weighted kinetic $\ell^2$ norm and connects weighted LBM stability theory \cite{ginzburg2005equilibrium,junk2009weighted} to block-encoding normalization. The change of variables preserves the discrete kinetic dynamics and, for uniformly positive weights and $\tau^*\geq1/2$, gives the scaled collision operator unit spectral norm. The remaining normalization factors are determined by ratios of equilibrium weights at the departure and arrival positions and times of streaming. These ratios equal one for uniform steady advection. Under periodic streaming, logarithmic derivatives of the weights along streaming trajectories control each normalization factor. A derivative bound uniform in space, time, and mesh resolution therefore keeps their product bounded as the mesh is refined at fixed physical time. Physical diffusive attenuation remains in the solution norm. %[L02,L03]

The resulting factorization yields a single-step block-encoding with two flag qubits and no uniform singular value amplification. We compose these steps coherently using a reusable failure counter \cite{Lin2023timeODE}, then obtain the terminal state through amplitude amplification or read out linear functionals without amplification. We also construct a history-state linear system with terminal padding from the same scaled operator to compare query bounds and spectral conditioning. The comparison specifies initial-state preparation and distinguishes coefficient access along a classical time schedule from access in a superposition of time indices. %[L04,L05,L06,L07,L09]

Numerical tests in one and two dimensions distinguish preservation of the discrete kinetic trajectories from continuum discretization error. In a controlled comparison of sequential block-encodings with and without scaling, the success probability for preparing the quantum state encoding the terminal distribution after 50 steps rises from $1.95\times10^{-4}$ to $8.76\times10^{-2}$, an approximately $450\times$ gain. Classical normalization diagnostics for an unsteady vortex confirm bounded cumulative normalization under refinement up to $N_t=16{,}384$ steps at fixed physical time. Noiseless simulations of circuits decomposed into elementary gates verify the implementation, while exact diagnostics quantify the spectral conditioning of the history-state matrix. %[L10,L12,L13,L14,L15]

Section~\ref{sec:LBM} defines the discrete kinetic model and output tasks, with equilibrium energy scaling established in Section~\ref{subsec:energyScaling}. Section~\ref{sec:ADEquantum} constructs coherent time marching, and Section~\ref{sec:QLSALBM} analyzes the history-state formulation. Section~\ref{sec:numerics} presents numerical verification, and Section~\ref{sec:conclusion} gives the conclusions. The appendices supply supporting proofs, gate-synthesis rules, and benchmark tables.
\par\endgroup %[L01,L03,L07,L09,L17]

\section{Discrete kinetic model and computational task}\label{sec:LBM}

We consider a passive scalar on a periodic domain $\Omega\subset\mathbb R^d$, with prescribed velocity $\bb u(\bb x,t)$ and constant diffusivity $D>0$: %[L01]
\begin{equation}\label{adv}
 \frac{\partial\phi}{\partial t}+\nabla\cdot(\bb u\phi)
 =D\nabla^2\phi,
\end{equation}
where $\phi(\bb x,t)$ denotes the transported scalar concentration and $\bb u(\bb x,t)$ may vary smoothly in space and time. %[L01]

\subsection{Discrete BGK kinetics and closed distribution recurrence}\label{sec:representation}

Coherent evolution over many time steps requires a state that determines the next kinetic update. For BGK collision, the full distribution supplies this state, including the non-equilibrium components that are absent from the scalar concentration. We therefore write collision and streaming as a closed linear recurrence on that distribution. Let $f_i(\bb x,t)$ denote the population associated with discrete velocity $\bb e_i$, $i=0,\ldots,q-1$, in a $d$-dimensional $\mathrm{D}d\mathrm{Q}q$ lattice model. We use D2Q5 ($q=5$), with $\bb e_0=(0,0)$, $\bb e_1=(1,0)$, $\bb e_2=(-1,0)$, $\bb e_3=(0,1)$, and $\bb e_4=(0,-1)$; \ref{app:discreteDetails} gives the stencil. The D1Q3 specialization ($q=3$, $\bb e_0=0$, $\bb e_{1,2}=\pm 1$, weights $w_0=2/3$, $w_{1,2}=1/6$) follows by restricting the index range to $i\in\{0,1,2\}$. %[L01]

Let $t_k=k\Delta t$ for $k=0,\ldots,N_t$ up to final time $T=N_t\Delta t$, and write $f_i^k(\bb x)=f_i(\bb x,t_k)$ and $\phi^k(\bb x)=\phi(\bb x,t_k)$. With lattice speed $c=\Delta x/\Delta t$ and physical link velocities $\bb c_i=c\bb e_i$, the BGK lattice Boltzmann equation reads %[L01]
\begin{equation}\label{collision}
	f_i^{k+1}(\bb{x}+\bb{c}_i\Delta t)=f_i^k(\bb{x})-\frac{\Delta t}{\tau}\left[f_i^k(\bb{x})-f_i^{\mathrm{eq},k}(\bb{x})\right],
\end{equation}
where $\tau$ is the relaxation time and $\tau^*=\tau/\Delta t$ denotes the dimensionless relaxation parameter used throughout. For the linear advection--diffusion equation (ADE) \eqref{adv} with prescribed velocity $\bb u(\bb x,t)$, the local equilibrium distribution depends linearly on the scalar moment $\phi^k$: %[L01]
\begin{equation}\label{equibrium}
	f_i^{\mathrm{eq},k}(\bb x)=w_i\phi^k(\bb x)\left[1+\frac{\bb c_i\cdot\bb u(\bb x,t_k)}{c_s^2}\right],
\end{equation}
with standard D2Q5 weights $w_0 = 2/6$, $w_1 = w_2 = w_3 = w_4 = 1/6$ and lattice sound speed $c_s = c/\sqrt{3} = (\Delta x/\Delta t)/\sqrt{3}$. Because collision conserves the zeroth moment, the macroscopic scalar concentration is recovered locally via %[L01]
\begin{equation}\label{phisum}
	\phi^k(\bb x)=\sum_{i=0}^{q-1} f_i^k(\bb x).
\end{equation}

On $N$ spatial nodes $\{\bb x_j\}_{j\in[N]}$, we stack the populations by direction into $\bb f^k=[\bb f_0^k;\ldots;\bb f_{q-1}^k]\in\mathbb R^{qN}$ and set $E_I=[I_N,\ldots,I_N]$ so that $\bb\phi^k=E_I\bb f^k$. Throughout, bold symbols denote unnormalized classical vectors, whereas $\ket v=\bb v/\|\bb v\|$ for $\bb v\neq\bb 0$ denotes a normalized quantum state. Norms are Euclidean vector norms or induced spectral norms unless marked as the Frobenius norm $\|\cdot\|_F$. %[L01]

The prescribed equilibrium coefficients and their diagonal node matrices are
\begin{equation}\label{equ:A_i}
\begin{aligned}
 a_i(t,\bb x)&=w_i\left(1+\frac{\bb c_i\cdot\bb u(\bb x,t)}{c_s^2}\right),
 \\
 A_i^k&=\operatorname{diag}\bigl(a_i(t_k,\bb x_0),\ldots,a_i(t_k,\bb x_{N-1})\bigr),
\end{aligned}
\end{equation}
and satisfy the partition-of-unity identities
\begin{equation}\label{asum}
 \sum_{i=0}^{q-1} a_i(t,\bb x)=1,
 \qquad \sum_{i=0}^{q-1} A_i^k=I_N.
\end{equation}
Setting $\mathcal A^k=\operatorname{diag}(A_0^k,\ldots,A_{q-1}^k)$, $\alpha=1-1/\tau^*$, and $\beta=1/\tau^*$, we represent periodic streaming by $P=\operatorname{diag}(P_0,\ldots,P_{q-1})$, where $P_i\ket j=\ket{\pi_i(j)}$ shifts $\bb x_j$ to $\bb x_j+\bb c_i\Delta t$. Combining collision and streaming (\ref{app:discreteDetails}) yields the closed single-step recurrence %[L01]
\begin{equation}\label{ftime}
  \bb f^{k+1}=M^k\bb f^k,\qquad
  M^k=P\bigl(\alpha I_{qN}+\beta\mathcal A^k E_I^\top E_I\bigr).
 \end{equation}
Equation~\eqref{ftime} determines the next distribution without reconstructing an intermediate macroscopic state. For $\tau^*\neq 1$, it retains the non-equilibrium flux components that $\bb\phi^k$ alone does not determine. Periodic boundary conditions make $P$ an exact permutation unitary; extending the construction to non-permutation wall closures would require a separate unitary dilation. %[L01]

\subsection{Continuum consistency scope and quantum output tasks}

The quantum solver reproduces the discrete recurrence~\eqref{ftime}. Its error relative to the continuum equation therefore contains the classical discretization error as well as any quantum approximation error. The BGK relation specifies the nominal continuum diffusivity as %[L16]
\begin{equation}\label{eq_D}
 D=c_s^2\left(\tau-\frac{\Delta t}{2}\right)
  =\frac13\left(\tau^*-\frac12\right)\frac{\Delta x^2}{\Delta t},
 \qquad \tau^*=\frac{\tau}{\Delta t}.
\end{equation}
The linear equilibrium~\eqref{equibrium} omits quadratic velocity terms, which introduces an advection-dependent consistency correction in the discrete recurrence. Lemma~\ref{lem:consistency} in \ref{app:consistencyProof} derives this correction in a formal modified equation at fixed lattice speed $c$. Under diffusive refinement $\Delta t=\gamma\Delta x^2$, the displayed leading correction is $\mathcal O(\Delta x^2)$ if the solution derivatives remain uniformly bounded. Extending this observation to second-order convergence of the complete scheme requires control of the Taylor remainder as $c\to\infty$, beyond the fixed-$c$ result. Section~\ref{subsec:exp4Unsteady} provides numerical convergence evidence for spatially uniform oscillatory advection. %[L16]

For the discrete recurrence, we consider three output tasks: preparing the terminal distribution state $\ket{f^{N_t}}$, preparing the terminal scalar state $\ket{\phi^{N_t}}$ when nonzero, and estimating a specified linear functional of $\bb\phi^{N_t}$. These tasks use the conditions on initial-state preparation and coefficient access in Sections~\ref{subsec:structuredBlockEncoding} and~\ref{subsec:readout}. Initial-state preparation and output extraction enter the resource analysis explicitly; recovering the complete classical field requires additional measurements.
\section{Equilibrium energy scaling and multistep normalization control}\label{subsec:energyScaling} %[L06]

A nonunitary update must be divided by a normalization factor before it can appear as a block of a unitary. Applying this normalization at every step can suppress the successful component of a coherent computation even when the classical recurrence is stable. We choose a quadratic metric based on the inverse equilibrium weights so that collision has unit spectral norm for $\tau^*\geq 1/2$ and the remaining normalization is determined by weight changes along streaming links. %[L02]

\subsection{Block-encoding normalization and the weighted representation}

The unscaled propagator $M^k$ conserves total scalar mass: $E_I M^k\bb f^k$ and $E_I\bb f^k$ have the same sum over spatial nodes. It is also $\ell^1$-contractive when $a_i\geq 0$ and $\tau^*\geq 1$. Neither property supplies the Euclidean spectral-norm bound needed for quantum block-encoding. In the over-relaxation interval $1/2\leq\tau^*<1$, even $\ell^1$ contractivity is not guaranteed. %[L02]

\begin{remark}[Oblique collision projection and normalization loss]\label{rem:weightedNorm}
At full relaxation ($\tau^*=1$, $\alpha=0$), the unscaled collision matrix at one node is the rank-one oblique projector $\bb a\bb 1^\top$, where $\bb a$ collects the local coefficients $a_i$. Its Euclidean spectral norm is $\sqrt{q}\|\bb a\|_2\geq 1$, with strict inequality for non-uniform $\bb a$. At zero velocity, this norm is $\sqrt{q}\|\bb a\|_2=\sqrt{3(4/9+2/36)}=\sqrt{3/2}\approx 1.2247$ in D1Q3 and $\sqrt{5(1/9+4/36)}=\sqrt{10/9}\approx 1.0541$ in D2Q5. Over-relaxation can increase the norm further: the D2Q5 example at $\tau^*=0.8$ has $\|M^0\|_2\approx 1.089$. The scaling introduced next converts the full-relaxation collision into an orthogonal projection of norm $1$, removing the normalization loss associated with this oblique projection.
\end{remark} %[L02]

To construct this metric, assume a grid-independent positivity margin $a_*>0$ such that, uniformly over all $(t,\bb x_j)$ under consideration, %[L02]
\begin{equation}\label{eq:positivity}
 a_i(t,\bb x_j)\geq a_*>0,\qquad i=0,\ldots,q-1,\quad j\in[N].
\end{equation}
For D2Q5 and D1Q3, the velocity bound $|u_\gamma|\leq(1/3-\delta_u)c$ along each coordinate axis $\gamma$, with $0<\delta_u<1/3$, ensures $a_i(t,\bb x_j)\geq\delta_u/2=:a_*>0$. Each $A_i^k$ is then strictly positive definite. Multiplication by the inverse square roots of the equilibrium weights defines the scaled distribution and its quadratic norm: %[L02]
\[
 \bb g^k=(\mathcal A^k)^{-1/2}\bb f^k,
 \qquad
 \|\bb g^k\|_2^2=\sum_{i=0}^{q-1}\sum_{j=0}^{N-1}\frac{|f_{i,j}^k|^2}{a_i(t_k,\bb x_j)}.
\]

Applying this change of variables at both time levels in~\eqref{ftime} gives
\begin{equation}\label{gtime}
 \bb g^{k+1}=B^k\bb g^k,\qquad
 B^k=(\mathcal A^{k+1})^{-1/2}M^k(\mathcal A^k)^{1/2}=\widetilde P^k\widetilde C^k,
\end{equation}
where the scaled streaming and collision factors are
\begin{align*}
 \widetilde P^k&=(\mathcal A^{k+1})^{-1/2}P(\mathcal A^k)^{1/2},\\
 \widetilde C^k&=\alpha I_{qN}+\beta\Pi^k,
 \qquad \Pi^k=Q^k(Q^k)^\dagger,\\
 Q^k&=(\mathcal A^k)^{1/2}E_I^\top=\begin{bmatrix}(A_0^k)^{1/2}\\\vdots\\(A_{q-1}^k)^{1/2}\end{bmatrix}\in\mathbb R^{qN\times N}.
\end{align*}
The columns of $Q^k$ are the local equilibrium states $\ket{\psi_{\mathrm{eq}}^k(\bb x_j)}=\sum_{i=0}^{q-1}\sqrt{a_i(t_k,\bb x_j)}\ket i$, each supported at a distinct node $\bb x_j$. Since the weights sum to one, these columns are orthonormal and $\Pi^k$ is an orthogonal projection. Thus $\widetilde C^k$ preserves the equilibrium component and multiplies its orthogonal complement by $\alpha$. Streaming then transfers amplitude from $(i,\bb x_j)$ to $(i,\bb x_{\pi_i(j)})$, with multiplier $\sqrt{a_i(t_k,\bb x_j)/a_i(t_{k+1},\bb x_{\pi_i(j)})}$. This factorization separates the collision contraction from changes in the weights between the departure and arrival points. %[L02]

\subsection{Collision projection and streaming norm}

The orthogonality of the collision projection fixes the norm of the scaled collision operator. The streaming norm is the square root of the largest ratio of weights along a streaming link, because streaming permutes the grid points. The following theorem combines these two properties to bound a complete step. %[L02]

\begin{theorem}\label{thm:norm}
 Let $P_i\ket j=\ket{\pi_i(j)}$, where each $\pi_i$ is a permutation on $[N]$. Under \eqref{eq:positivity} and $\tau^*\geq1/2$, for every time step $k$, %[L02]
 \[
 \|\widetilde C^k\|_2=1,\qquad
 \|B^k\|_2\leq\|\widetilde P^k\|_2=\rho^k,
 \]
 where
 \begin{equation}\label{eq:coefficientRatioBound}
 \rho^k=\max_{i,j}
 \sqrt{\frac{a_i(t_k,\bb x_j)}{a_i(t_{k+1},\bb x_{\pi_i(j)})}}.
 \end{equation}
\end{theorem}
\begin{proof}
 By \eqref{asum}, $(Q^k)^\dagger Q^k=\sum_{i=0}^{q-1} A_i^k=I_N$, so $(\Pi^k)^2=Q^k(Q^k)^\dagger Q^k(Q^k)^\dagger=\Pi^k$ and $(\Pi^k)^\dagger=\Pi^k$; hence $\Pi^k$ is an orthogonal projection. Since $\beta=1-\alpha$, %[L02]
 \[
 \widetilde C^k=\Pi^k+\alpha(I_{qN}-\Pi^k)
 \]
 is self-adjoint with eigenvalues $1$ on $\operatorname{Im}(\Pi^k)$ and $\alpha=1-1/\tau^*\in[-1,1)$ on $\ker(\Pi^k)$, proving $\|\widetilde C^k\|_2=1$. For each direction $i$, $\widetilde P_i^k=(A_i^{k+1})^{-1/2}P_i(A_i^k)^{1/2}$ maps the orthonormal basis $\{\ket j\}_{j=0}^{N-1}$ to orthogonal vectors $\widetilde P_i^k\ket j=\sqrt{a_i(t_k,\bb x_j)/a_i(t_{k+1},\bb x_{\pi_i(j)})}\ket{\pi_i(j)}$, so $\|\widetilde P_i^k\|_2=\max_j\sqrt{a_i(t_k,\bb x_j)/a_i(t_{k+1},\bb x_{\pi_i(j)})}$. Taking the maximum over $i\in\{0,\ldots,q-1\}$ gives $\|\widetilde P^k\|_2=\rho^k$, and submultiplicativity yields $\|B^k\|_2\leq\|\widetilde P^k\|_2\|\widetilde C^k\|_2=\rho^k$.
\end{proof} %[L02]

\begin{remark}[Relaxation boundary $\tau^*=1/2$]
Classical BGK-LBM stability requires $\tau^*\geq 1/2$ \cite[\S3.5.3.1]{kruger2017lattice}, which corresponds directly to $|\alpha|=|1-1/\tau^*|\leq 1$. Theorem~\ref{thm:norm} covers the entire interval $\tau^*\geq 1/2$, including the non-dissipative reflection boundary $\tau^*=1/2$ ($\alpha=-1$, $D=0$) where $\widetilde C^k=2\Pi^k-I$ is strictly unitary. Physical ADE simulations with $D>0$ operate in the diffusive interior $\tau^*>1/2$.
\end{remark} %[L02]

\subsection{Finite-time bounds for space- and time-dependent coefficients}

For the single-step bound to control a many-step computation, its product must remain bounded as $N_t$ increases at fixed physical time $T$. For spatially uniform steady advection, $a_i(t_k,\bb x_j)=a_i(t_{k+1},\bb x_{\pi_i(j)})$, so $\rho^k=1$ and $\|B^k\|_2\leq 1$. For variable velocity on a periodic domain, $\bb x_{\pi_i(j)}=\bb x_j+\bb c_i\Delta t\pmod{\bb L}$. The two weights in each ratio are therefore linked by the derivative $\mathscr D_i=\partial_t+\bb c_i\cdot\nabla$ along that streaming trajectory. A bound on $\mathscr D_i\log a_i$ controls the accumulated normalization. %[L03]

\begin{lemma}[Grid-independent growth bound for unsteady velocity fields]\label{lem:spaceTimeGrowth}
 Under the assumptions of Theorem~\ref{thm:norm}, suppose streaming is periodic and $a_i\in C^1([0,T]\times\Omega)$ satisfies \eqref{eq:positivity} and %[L03]
 \begin{equation}\label{eq:spaceTimeSmoothness}
  \max_{0\leq i<q}\sup_{(t,\bb x)\in[0,T]\times\Omega}
  \left|(\partial_t+\bb c_i\cdot\nabla)\log a_i(t,\bb x)\right|
  \leq\ell_{\mathrm{st}}
 \end{equation}
 for a grid-independent constant $\ell_{\mathrm{st}}\geq0$. Then, for all $0\leq k<N_t$ and $0\leq r<s\leq N_t$, %[L03]
 \begin{equation}\label{eq:timeLevelGrowth}
  \|B^k\|_2\leq\rho^k\leq\exp(\ell_{\mathrm{st}}\Delta t/2),
 \end{equation}
 \begin{equation}\label{eq:timeIntervalGrowth}
  \|B^{s-1}\cdots B^r\|_2
  \leq\prod_{k=r}^{s-1}\rho^k
  \leq\exp\bigl(\ell_{\mathrm{st}}(t_s-t_r)/2\bigr).
 \end{equation}
 If $0\leq\ell_{\mathrm{st}}\Delta t\leq1$, then $\rho^k\leq 1+\ell_{\mathrm{st}}\Delta t$.
\end{lemma}
\begin{proof}
 Along the characteristic segment $s\mapsto(t_k+s,\bb x_j+s\bb c_i)$ for $s\in[0,\Delta t]$, the fundamental theorem of calculus gives %[L03]
 \[
  \left|\log\frac{a_i(t_{k+1},\bb x_{\pi_i(j)})}{a_i(t_k,\bb x_j)}\right|
  =\left|\int_0^{\Delta t}\mathscr D_i\log a_i(t_k+s,\bb x_j+s\bb c_i)\,\d s\right|
  \leq\ell_{\mathrm{st}}\Delta t.
 \]
 Exponentiating, taking square roots, and maximizing over $(i,j)$ yields \eqref{eq:timeLevelGrowth}; submultiplicativity gives \eqref{eq:timeIntervalGrowth}, and $e^{x/2}\leq 1+x$ on $[0,1]$ yields the linear bound.
\end{proof} %[L03]

For D1Q3 and D2Q5, \ref{app:coefficientBounds} expresses $\ell_{\mathrm{st}}$ in terms of bounds on the prescribed velocity derivatives, allowing it to be computed without scanning the grid. Under the diffusive refinement used in Section~\ref{subsec:exp4Unsteady}, uniform bounds on those derivatives and a fixed positive lower bound on the weights keep $\ell_{\mathrm{st}}$ uniform across grids. Lemma~\ref{lem:spaceTimeGrowth} therefore permits the same normalization factor to be used at every step: %[L03]
\begin{equation}\label{eq:commonNormalizationBound}
 \Lambda\geq\max\{1,\rho^0,\ldots,\rho^{N_t-1}\},
\end{equation}
by choosing $\Lambda=\exp(\ell_{\mathrm{st}}\Delta t/2)$. For spatially uniform steady flow, $\Lambda=1$ suffices.

Three quantities connect this normalization bound to the quantum resource analysis: the cumulative coefficient-growth factor, the scaled-solution norm ratio, and their product, which determines the amplification cost: %[L03]
\begin{equation}\label{eq:quantumGrowth}
 \Gamma_T=\Lambda^{N_t},\qquad
 R_T=\frac{\|\bb g^0\|}{\|\bb g^{N_t}\|},\qquad K_T=\Gamma_T R_T\geq 1,
\end{equation}
Here $\bb g^0,\bb g^{N_t}\neq\bb 0$. Under Lemma~\ref{lem:spaceTimeGrowth}, the choice $\Lambda=\exp(\ell_{\mathrm{st}}\Delta t/2)$ gives $\Gamma_T=\exp(\ell_{\mathrm{st}}T/2)$, independent of $N_t$ and $N$ at fixed $T$ when $\ell_{\mathrm{st}}$ is uniform across grids. Uniform steady flow gives $\Lambda=\Gamma_T=1$. The factor $\Gamma_T$ records the normalization chosen from coefficient bounds, whereas $R_T$ records the actual change in the scaled distribution norm, including attenuation caused by physical diffusion. Their product $K_T$ governs the success probability derived in Section~\ref{sec:ADEquantum}. %[L03]

\section{Coherent quantum time marching and observable readout}
\label{sec:ADEquantum}

The factorization $B^k=\widetilde P^k\widetilde C^k$ supplies the collision and streaming factors used to construct a single-step block-encoding. We compose $N_t$ steps through a reusable failure counter and derive bounds for state recovery, observable readout, and query complexity. %[L04,L05,L06,L07]

\subsection{Structured single-step block-encoding}\label{subsec:structuredBlockEncoding}

The single-step normalization follows directly from the two factors in $B^k=\widetilde P^k\widetilde C^k$. Under positive weights~\eqref{eq:positivity}, $\tau^*\geq 1/2$, and periodic streaming, Theorem~\ref{thm:norm} gives $\|\widetilde C^k\|_2=1$ for collision and $\|\widetilde P^k\|_2=\rho^k\leq\Lambda$ for weighted streaming. We therefore embed collision with unit normalization and streaming with normalization $\Lambda$, using two independent flag qubits whose zero states select the respective blocks. The jointly selected block is $B^k/\Lambda$, so uniform singular value amplification is unnecessary. The complete construction also allocates the failure counter for multistep evolution, a recovery flag for terminal distribution recovery, and reversible work qubits when the arithmetic access model is used. %[L04]

The registers and padding rules implement this normalized factorization. Recall that a unitary $U_B$ is an $(\alpha_B,a,\varepsilon_B)$-block-encoding of $B$ if $\|B-\alpha_B(\bra{0^a}\otimes I)U_B(\ket{0^a}\otimes I)\|_2\leq\varepsilon_B$ \cite{Gilyen2019QSVD,Chakraborty2019blockEncode}. We encode the velocity index $i\in\{0,\ldots,q-1\}$ in $n_d=\lceil\log_2 q\rceil$ direction qubits ($n_d=2$ for D1Q3, $n_d=3$ for D2Q5) and the spatial node $j\in[N]$ in $n_g=\lceil\log_2 N\rceil$ grid qubits, taking $N=2^{n_g}$ without loss of generality. On unused direction states ($i\geq q$), we define the single-step target $B^k$ to act as zero. Specifically, the collision factor extends by $\alpha I$ on unused directions while the streaming flag enforces the zero product block, using the unitary streaming extension $P_{\mathrm{ext}}=P\oplus I_{\mathrm{unused}}$. We consider two implementations of coefficient access:
\begin{itemize}
 \item \textbf{Coherent arithmetic oracle model:} We assume that the coefficients in $\bb u(\bb x,t)$ and the resulting $a_i(t_k,\bb x_j)$ admit an efficient reversible evaluation. Under this access assumption, a reversible value oracle $O_a$ computes $b$-bit approximations of $a_i(t_k,\bb x_j)$ on input $(k,i,j)$ using $\mathcal O(1)$ queries and $\mathcal O(\operatorname{poly}(b,\log N))$ reversible arithmetic gates, and uncomputes its scratch wires after rotation synthesis.
 \item \textbf{Direct multiplexed rotation compilation:} For explicit circuit synthesis without arithmetic work qubits, node-dependent angles compile directly into uniformly controlled $R_Y$ multiplexors \cite{barenco1995elementary,shende2006synthesis} costing $\mathcal O(qN)$ elementary gates per step, drawn from single-qubit and controlled-NOT gates (\texttt{\{1Q, CNOT\}}). Section~\ref{sec:numerics} and \ref{app:circuitSynthesis} implement this compilation at the elementary-gate level.
\end{itemize}
Periodic Cartesian streaming $P$ can be realized by direction-controlled modular increment and decrement circuits on the coordinate subregisters of $\ket j$ with $\mathcal O(\operatorname{poly}(n_g))$ gates \cite{todorova2020quantum,goubault2020quantum}. This efficient modular-arithmetic implementation is used in the oracle resource model. The directly compiled circuits instead expand multi-controlled gates without arithmetic work qubits; their gate counts are reported with the module costs in the appendix. %[L04]

To encode collision, let $V^k$ be a block-diagonal local state-preparation unitary acting on $\ket{\cdot}_d\ket j_g$ such that %[L04]
\begin{equation}\label{eq:localEquilibriumPreparation}
 V^k\ket{0^{n_d}}_d\ket j_g
 =\sum_{i=0}^{q-1}\sqrt{a_i(t_k,\bb x_j)}\ket i_d\ket j_g
 =Q^k\ket j_g,
\end{equation}
while leaving the unused direction subspace $\operatorname{span}\{\ket i_d:q\leq i<2^{n_d}\}$ invariant. Explicit two- and three-layer $R_Y$-multiplexor constructions for D1Q3 and D2Q5 are detailed in \ref{app:circuitSynthesis}. Here $\sum_{i=0}^{q-1}a_i(t_k,\bb x_j)=1$ normalizes each local target state, and the unitary construction of $V^k$ generates the orthogonal reflection %[L04]
\[
 R_\Pi^k:=V^k\bigl((2\ket{0^{n_d}}\bra{0^{n_d}}-I_{2^{n_d}})\otimes I_N\bigr)(V^k)^\dagger,
\]
which equals $2\Pi^k-I_{qN}$ on the physical subspace and $-I$ on unused directions. Since $\alpha=1-1/\tau^*\in[-1,1)$, the decomposition %[L04]
\begin{equation}\label{eq:collisionLCU}
 \widetilde C^k=\Pi^k+\alpha(I_{qN}-\Pi^k)
 =\frac{1+\alpha}{2}I_{qN}+\frac{1-\alpha}{2}R_\Pi^k
\end{equation}
expresses $\widetilde C^k$ on the physical subspace as a convex combination of two unitaries with weights $(1+\alpha)/2=1-1/(2\tau^*)\in[0,1)$ and $(1-\alpha)/2=1/(2\tau^*)\in(0,1]$. Preparing a single collision flag $\ket{0}_{a_C}$ with angle $\theta_\alpha=2\arcsin\sqrt{1/(2\tau^*)}$ therefore yields an exact $(1,1,0)$-block-encoding $U_C^k$ whose success block is $\widetilde C^k$ on physical directions and $\alpha I$ on unused directions. As shown in \ref{app:circuitSynthesis}, the compiled gate sequence requires no ancilla control on $V^k$ or $(V^k)^\dagger$. %[L04]

To encode streaming, we define the space-time link ratios and diagonal scaling operator
\begin{equation}\label{eq:streamingRatio}
 r_{ij}^k=\sqrt{\frac{a_i(t_k,\bb x_j)}{a_i(t_{k+1},\bb x_{\pi_i(j)})}}\leq\rho^k\leq\Lambda,
 \qquad
 D_r^k=\sum_{i=0}^{q-1}\sum_{j=0}^{N-1}r_{ij}^k\ket{i,j}\bra{i,j}.
\end{equation}
Let $D_\Lambda^k$ denote the controlled-rotation unitary acting on a second flag ancilla $\ket{0}_{a_P}$ via $R_Y\bigl(2\arccos(r_{ij}^k/\Lambda)\bigr)$ for $0\leq i<q$ and angle $\pi$ on unused directions $i\geq q$. Setting the angle to $\pi$ assigns zero $\ket 0_{a_P}$-amplitude to unused velocity states, so composing $D_\Lambda^k$ with $I_2\otimes P_{\mathrm{ext}}$ yields an exact $(\Lambda,1,0)$-block-encoding $U_P^k$ of $\widetilde P^k=PD_r^k$. For spatially uniform steady flow ($\Lambda=1$, $r_{ij}^k\equiv 1$), the physical ratio rotations reduce to the identity while unused directions still require zeroing in the successful block. With $\Pi_{\mathrm{phys}}=\sum_{i=0}^{q-1}\ket i\bra i\otimes I_N$, the reduced circuit takes the form %[L04]
\[
 F_{\mathrm{pad}}=I_{a_P}\otimes\Pi_{\mathrm{phys}}
       +X_{a_P}\otimes(I-\Pi_{\mathrm{phys}}),\qquad
 U_{P,\mathrm{unif}}=(I_2\otimes P_{\mathrm{ext}})F_{\mathrm{pad}}.
\]
Replacing the node-dependent ratio multiplexor by the direction-only padding flip $F_{\mathrm{pad}}$ removes all unused velocity states ($i\geq q$) from the zero-flag block, so that $\bra 0_{a_P}U_{P,\mathrm{unif}}\ket 0_{a_P}=P\oplus 0$ holds identically. %[L04]

\begin{proposition}\label{prop:structuredBE}
 Under the assumptions on coefficient access in Section~\ref{subsec:structuredBlockEncoding} and the bound in \eqref{eq:commonNormalizationBound}, the composition $U_B^k=(U_P^k)_{a_P,d,g}(U_C^k)_{a_C,d,g}$ on two flag ancillas $(a_P,a_C)$ and $(n_d+n_g)$ system qubits is an exact $(\Lambda,2,0)$-block-encoding of $B^k$ in the exact-angle model, using $\mathcal O(1)$ coefficient and permutation calls. At finite $b$-bit arithmetic precision with $b=\mathcal O(\log(1/\eta))$, the encoded block approximates $B^k/\Lambda$ within spectral-norm error $\eta$ uniformly in $k$.
\end{proposition}
Because $(a_P,a_C)$ are independent flag qubits, their joint zero-flag block is the product of the two selected blocks and requires no uniform singular value amplification \cite{Gilyen2019QSVD}. Moreover, because the weights in the linear combination of unitaries (LCU) for collision in~\eqref{eq:collisionLCU} sum to unity, the exact single-step normalization remains $\Lambda$ (proof and gate realization in \ref{app:circuitSynthesis}). %[L04]

\subsection{Coherent composition with a reusable failure counter}\label{sec:timemarching}\label{subsec:dilatingUnitarization}

Composing $N_t$ block-encoded steps coherently without intermediate postselection or $\mathcal O(N_t)$ flag ancillas requires isolating failed branches in an auxiliary counter register. For $j=1,\ldots,N_t$, let $U_j=U_B^{j-1}$ encode $B^{j-1}/\alpha_j$ with $a=2$ and $\alpha_j=\Lambda$. A reusable failure counter \cite{Lin2023timeODE} replaces $aN_t$ separate step flags with $a+n_c$ qubits, where $n_c=\lceil\log_2(N_t+1)\rceil$. On this counter register, define $\mathsf{ADD}\ket c=\ket{c+1\pmod{2^{n_c}}}$ and the conditional increment unitary %[L05]
\begin{equation}\label{eq:failureCounter}
 S=I_c\otimes\ket{0^a}\bra{0^a}_a\otimes I
   +\mathsf{ADD}\otimes(I_a-\ket{0^a}\bra{0^a}_a)\otimes I,
\end{equation}
synthesized as an unconditional $\mathsf{ADD}$ followed by a zero-flag-controlled decrement $\mathsf{ADD}^\dagger$. The implemented multistep circuit applies each block-encoding without a counter control through $\widetilde{\mathcal V}_j=S(I_c\otimes U_j)$, which avoids multi-controlled gate overhead on $U_j$. For the proof, it is useful to compare this circuit with a reference step applied only when the counter is zero: %[L05]
\[
 \overline U_j=\ket{0^{n_c}}\bra{0^{n_c}}_c\otimes U_j+(I_c-\ket{0^{n_c}}\bra{0^{n_c}}_c)\otimes I_{a,\mathrm{sys}},
 \qquad \mathcal V_j=S\overline U_j,
\]
as illustrated in Fig.~\ref{fig:TimeRegisters}. Under $\mathcal V_j$, the counter sectors follow the explicit staircase evolution~\eqref{flowmap} in \ref{app:counterProof} with $\nu_j=\prod_{r=1}^j\alpha_r$. Proposition~\ref{prop:counterInvariance} establishes equality of the projected zero-counter success blocks of the implemented and reference circuits, which justifies using the simpler implementation to extract the terminal state. %[L05]

\begin{figure*}[!htb]
 \centering
 \resizebox{0.90\textwidth}{!}{%
\begin{tikzpicture}[
   x=0.77cm, y=0.90cm,
   wire/.style={line width=0.65pt, draw=black!85},
   ctrlwire/.style={line width=0.65pt, draw=black!85},
   gate/.style={draw=black!85, line width=0.65pt, fill=white, rounded corners=1.2pt, minimum height=0.60cm, minimum width=0.82cm, font=\small},
   multigate/.style={draw=black!85, line width=0.7pt, fill=white, rounded corners=1.5pt, minimum width=0.96cm, font=\small},
   stagebox/.style={draw=blue!55!black, line width=0.6pt, densely dashed, fill=blue!4!white, rounded corners=2pt},
   sbox/.style={draw=teal!65!black, line width=0.55pt, densely dotted, fill=teal!4!white, rounded corners=1.5pt},
   projbox/.style={draw=black!75, line width=0.6pt, fill=gray!10, rounded corners=1.2pt, minimum height=0.52cm, inner xsep=3pt, font=\footnotesize},
   ocircle/.style={draw=black!90, line width=0.65pt, fill=white, circle, inner sep=0pt, minimum size=4.6pt},
   refcircle/.style={draw=red!70!black, line width=0.75pt, fill=white, circle, inner sep=0pt, minimum size=5.0pt}
 ]
   % Stage background boxes for V_1 and V_{N_t}
   \draw[stagebox] (1.45,-2.52) rectangle (5.65,0.68);
   \node[font=\footnotesize, text=blue!55!black, anchor=south] at (3.55,0.70) {Step $1$: $\mathcal{V}_1 = S\overline{U}_1$};
   \draw[sbox] (3.00,-1.42) rectangle (5.48,0.48);
   \node[font=\scriptsize, text=teal!60!black, anchor=south west, inner sep=1.5pt] at (3.06,-1.38) {$S$};

   \draw[stagebox] (7.25,-2.52) rectangle (11.45,0.68);
   \node[font=\footnotesize, text=blue!55!black, anchor=south] at (9.35,0.70) {Step $N_t$: $\mathcal{V}_{N_t} = S\overline{U}_{N_t}$};
   \draw[sbox] (8.80,-1.42) rectangle (11.28,0.48);
   \node[font=\scriptsize, text=teal!60!black, anchor=south west, inner sep=1.5pt] at (8.86,-1.38) {$S$};

   % Horizontal quantum register wires
   \foreach \y in {0,-1,-2} {
     \draw[wire] (0.55,\y) -- (6.05,\y);
     \draw[wire] (6.85,\y) -- (12.05,\y);
   }
   % Center ellipsis
   \node[font=\normalsize] at (6.45,0) {$\cdots$};
   \node[font=\normalsize] at (6.45,-1) {$\cdots$};
   \node[font=\normalsize] at (6.45,-2) {$\cdots$};

   % Register labels on left
   \node[anchor=east, font=\small] at (0.50,0) {Counter $\ket{0^{n_c}}_c$};
   \node[anchor=east, font=\small] at (0.50,-1) {Ancilla $\ket{0^a}_a$};
   \node[anchor=east, font=\small] at (0.50,-2) {State $\ket{g^0}_{d,g}$};

   % Multi-qubit bundle slashes
   \draw[line width=0.6pt] (0.85,-0.13) -- (1.03,0.13);
   \node[font=\scriptsize, anchor=south] at (0.94,0.10) {$n_c$};
   \draw[line width=0.6pt] (0.85,-1.13) -- (1.03,-0.87);
   \node[font=\scriptsize, anchor=south] at (0.94,-0.90) {$a$};
   \draw[line width=0.6pt] (0.85,-2.13) -- (1.03,-1.87);
   \node[font=\scriptsize, anchor=south] at (0.94,-1.90) {$n_d{+}n_g$};

   % Step 1 gates
   \draw[ctrlwire, draw=red!70!black, densely dashed] (2.20,0) -- (2.20,-1.0);
   \node[refcircle] at (2.20,0) {};
   \draw[multigate] (1.66,-2.34) rectangle (2.74,-0.66);
   \node[font=\small] at (2.20,-1.50) {$U_1$};

   \node[gate] (add1) at (3.62,0) {$\mathsf{ADD}$};
   \draw[ctrlwire] (4.86,0) -- (4.86,-1);
   \node[ocircle] at (4.86,-1) {};
   \node[gate] (addd1) at (4.86,0) {$\mathsf{ADD}^\dagger$};

   % Step N_t gates
   \draw[ctrlwire, draw=red!70!black, densely dashed] (8.00,0) -- (8.00,-1.0);
   \node[refcircle] at (8.00,0) {};
   \draw[multigate] (7.46,-2.34) rectangle (8.54,-0.66);
   \node[font=\small] at (8.00,-1.50) {$U_{N_t}$};

   \node[gate] (addN) at (9.42,0) {$\mathsf{ADD}$};
   \draw[ctrlwire] (10.66,0) -- (10.66,-1);
   \node[ocircle] at (10.66,-1) {};
   \node[gate] (adddN) at (10.66,0) {$\mathsf{ADD}^\dagger$};

   % Terminal projections / output annotation
   \node[projbox, anchor=west] at (12.05,0) {$\bra{0^{n_c}}_c$};
   \node[projbox, anchor=west] at (12.05,-1) {$\bra{0^a}_a$};
   \node[anchor=west, font=\small] at (12.05,-2) {$\propto B^{N_t-1}\cdots B^0\ket{g^0}_{d,g}$};

   % Legend note for reference vs implemented circuit
   \node[refcircle] at (2.35,-2.92) {};
   \node[anchor=west, font=\footnotesize, text=black!80] at (2.52,-2.92) {Zero-counter control in reference step $\mathcal{V}_j=S\overline{U}_j$ (omitted in implemented step $\widetilde{\mathcal{V}}_j=S(I_c\otimes U_j)$)};
 \end{tikzpicture}%
}
 \caption{Reference circuit used in the proof of dilating unitarization with one reusable $a$-qubit flag register and an $n_c=\lceil\log_2(N_t+1)\rceil$-qubit failure counter. The controlled construction shown defines $\mathcal V_j=S\overline U_j$. The implemented circuit omits the open zero-counter controls and uses $\widetilde{\mathcal V}_j=S(I_c\otimes U_j)$. Proposition~\ref{prop:counterInvariance} establishes equality of their projected $(\ket{0^{n_c}}_c\ket{0^a}_a)$ success blocks.}\label{fig:TimeRegisters}
\end{figure*}

\begin{proposition}[Exact zero-counter block extraction]\label{prop:counterInvariance}
 Let $N_t\geq 1$, $n_c=\lceil\log_2(N_t+1)\rceil$, and let $U_j$ be an $(\alpha_j,a,0)$-block-encoding of $B^{j-1}$ for $j=1,\ldots,N_t$. Define $\mathcal U_{N_t}=\mathcal V_{N_t}\cdots\mathcal V_1$ and $\widetilde{\mathcal U}_{N_t}=\widetilde{\mathcal V}_{N_t}\cdots\widetilde{\mathcal V}_1$ with $\mathcal V_j=S\overline U_j$ and $\widetilde{\mathcal V}_j=S(I_c\otimes U_j)$. Then both unitaries satisfy %[L05]
 \begin{equation}\label{eq:productBlock}
\begin{aligned}
  & (\bra{0^{n_c}}_c\bra{0^a}_a\otimes I)\mathcal U_{N_t}(\ket{0^{n_c}}_c\ket{0^a}_a\otimes I)
  \\ &=(\bra{0^{n_c}}_c\bra{0^a}_a\otimes I)\widetilde{\mathcal U}_{N_t}(\ket{0^{n_c}}_c\ket{0^a}_a\otimes I)
  \\ &=\frac{B^{N_t-1}\cdots B^0}{\alpha_{N_t}\cdots\alpha_1}.
 \end{aligned}
\end{equation}
\end{proposition}
Each failed branch acquires a positive counter value and receives at most one increment per remaining step. Because $2^{n_c}>N_t$, no failed branch can wrap modulo $2^{n_c}$ back to $c=0$ within $N_t$ steps. This non-wraparound property underlies exact zero-counter extraction, and \ref{app:counterProof} proves the induction for both controlled and uncontrolled step unitaries. %[L05]

By Proposition~\ref{prop:counterInvariance}, preparing $\ket{0^{n_c}}_c\ket{0^a}_a\ket{g^0}$ and applying $\mathcal U_{N_t}$ or $\widetilde{\mathcal U}_{N_t}$ succeeds upon measuring $c=0$ and $a=0^a$ with exact probability %[L05]
\begin{equation}\label{Pru0ut}
 \mathbb{P}_{\mathrm{succ}}=
 \frac{1}{(\alpha_{N_t}\cdots\alpha_1)^2}
 \frac{\|\bb g^{N_t}\|^2}{\|\bb g^0\|^2}.
\end{equation}
With $\alpha_j=\Lambda$, Eq.~\eqref{Pru0ut} gives the time-marching success probability $\mathbb P_{\mathrm{TM}}=K_T^{-2}$. Our circuit implementations use the product of uncontrolled steps $\widetilde{\mathcal U}_{N_t}$. Incorporating the terminal readout and precision rules derived next yields the complete solver summarized in Algorithm~\ref{alg:timeMarching} in Section~\ref{subsec:growthAccuracy}. %[L05]

\subsection{Initial preparation, terminal recovery, and macroscopic readout}\label{subsec:readout}

The transformations between the physical kinetic distribution and its energy-scaled representation determine the input preparation cost and the success probability of terminal readout. For equilibrium initialization $\bb f^0=\mathcal A^0 E_I^\top\bb\phi^0$, the scaled vector satisfies $\bb g^0=Q^0\bb\phi^0$ with $\|\bb g^0\|=\|\bb\phi^0\|$, so preparing $\ket{\phi^0}=U_{\phi^0}\ket{0^{n_g}}$ followed by $V^0$ prepares $\ket{g^0}$ deterministically with unit probability. For general non-equilibrium initial data $\ket{f^0}$, such as the flux-perturbation test in Section~\ref{sec:numerics}, a controlled rotation encoding $\sqrt{a_*}(\mathcal A^0)^{-1/2}$ prepares $\ket{g^0}$ with success probability $\|\bb g^0\|^2 a_*/\|\bb f^0\|^2\geq a_*$. The terminal circuit then serves three distinct tasks: recovery of the full kinetic distribution state, preparation of the normalized macroscopic scalar state, and direct estimation of a prescribed linear functional. Their success probabilities and query costs are different and are treated separately below. %[L06]

\begin{lemma}[Recovery of the physical distribution state]\label{lem:recoverF}
 Under \eqref{eq:positivity} and \eqref{asum}, $\sqrt{a_*}I\preceq(\mathcal A^{N_t})^{1/2}\preceq I$ on the physical subspace, and a single controlled rotation realizes a $(1,1,0)$-block-encoding $U_{\mathrm{rec}}$ of $(\mathcal A^{N_t})^{1/2}$ using $\mathcal O(1)$ coefficient queries. Applied to $\ket{g^{N_t}}$, the zero-flag outcome prepares $\ket{f^{N_t}}$ with probability %[L06]
 \begin{equation}\label{eq:distributionRecovery}
  \mathbb{P}_{\mathrm{rec}}=\frac{\|\bb f^{N_t}\|^2}{\|\bb g^{N_t}\|^2}\geq a_*.
 \end{equation}
 If $\ket{\widehat g^{N_t}}$ satisfies $\|\ket{\widehat g^{N_t}}-\ket{g^{N_t}}\|\leq\varepsilon_g$ and $U_{\mathrm{rec}}$ has block error $\eta_{\mathrm{rec}}$, the normalized recovered state error is at most $2(\varepsilon_g+\eta_{\mathrm{rec}})/\sqrt{a_*}$.
\end{lemma}
Lemma~\ref{lem:recoverF} is proved in \ref{app:tmProof}. When macroscopic scalar output is sought instead of the full distribution, $\bb\phi^{N_t}$ follows directly from $\bb g^{N_t}$ via %[L06]
\begin{equation}\label{eq:macroExtraction}
 \bb\phi^{N_t}=E_I(\mathcal A^{N_t})^{1/2}\bb g^{N_t}=(Q^{N_t})^\dagger\bb g^{N_t}.
\end{equation}
Applying $(V^{N_t})^\dagger$ to $\ket{g^{N_t}}$ and projecting the direction register onto $\ket{0^{n_d}}_d$ therefore prepares $\ket{\phi^{N_t}}$ directly on the spatial register with conditional probability $\mathbb P_\phi=\|\bb\phi^{N_t}\|^2/\|\bb g^{N_t}\|^2$. %[L06]

Full-distribution recovery and scalar projection have different probability lower bounds. For arbitrary non-equilibrium states, equilibrium-weight positivity alone provides no positive lower bound on $\mathbb P_\phi$ because the zeroth moment may vanish even when $\bb f^{N_t}\neq\bb 0$. Combining counter evolution with scalar projection gives the pre-amplification probability %[L06]
\[
 \mathbb P_{\mathrm{TM},\phi}=\mathbb P_{\mathrm{TM}}\mathbb P_\phi
 =\frac{\mathbb P_\phi}{K_T^2}.
\]
Algorithm~\ref{alg:timeMarching} accordingly requires a certified positive lower bound $p_{\mathrm{out}}\leq\mathbb P_\phi$ in scalar mode, whereas distribution mode always admits $p_{\mathrm{out}}=a_*$. If the normalized scaled state carries error $\varepsilon_g$ and the terminal projection block approximates $(Q^{N_t})^\dagger$ to operator-norm error $\eta_\phi$, normalization yields %[L06]
\[
 \bigl\|\ket{\widetilde\phi^{N_t}}-\ket{\phi^{N_t}}\bigr\|
 \leq\frac{2(\varepsilon_g+\eta_\phi)}{\sqrt{p_{\mathrm{out}}}}.
\]
Choosing $\varepsilon_g,\eta_\phi\leq\varepsilon\sqrt{p_{\mathrm{out}}}/8$ alongside the step precision in Algorithm~\ref{alg:timeMarching} keeps the terminal state error within $\varepsilon$ at pre-amplification probability $\Omega(p_{\mathrm{out}}/K_T^2)$. Preparing the normalized scalar state therefore depends explicitly on the supplied lower bound $p_{\mathrm{out}}$. %[L06]

When the target output is a spatial linear functional $\mathcal J(\bb\phi^{N_t})=\bb w^\top\bb\phi^{N_t}$, such as subdomain mass, spatial moments, or low-wavenumber Fourier modes, preparing the normalized state $\ket{\phi^{N_t}}$ is unnecessary whenever $\ket w=\bb w/\|\bb w\|=U_w\ket{0^{n_g}}_g$ is efficiently preparable. For example, a half-domain mass indicator requires only $n_g-1$ Hadamard gates. Under equilibrium initialization ($\|\bb g^0\|=\|\bb\phi^0\|$), we form the unamplified end-to-end circuit $\mathcal W_\phi$ by composing equilibrium preparation, unamplified counter evolution, scalar projection, and $U_w^\dagger$ (explicit product in \ref{app:readoutCircuit}). By Proposition~\ref{prop:counterInvariance} and \eqref{eq:macroExtraction}, its all-zero diagonal matrix element satisfies %[L06]
 \begin{equation}\label{eq:linearObservableOverlap}
  \bra{0^{\mathrm{all}}}\mathcal W_\phi\ket{0^{\mathrm{all}}}
  =\frac{\bb w^\top\bb\phi^{N_t}}{\Gamma_T\|\bb w\|_2\|\bb\phi^0\|_2}.
 \end{equation}
Consequently, $\bb w^\top\bb\phi^{N_t}$ can be estimated within additive error $\varepsilon_{\mathrm{obs}}\|\bb w\|_2\|\bb\phi^0\|_2$ using $\mathcal O(\Gamma_T^2/\varepsilon_{\mathrm{obs}}^2)$ shots of a single-ancilla Hadamard test on $\mathcal W_\phi$, or $\mathcal O(\Gamma_T/\varepsilon_{\mathrm{obs}})$ coherent queries via quantum amplitude estimation \cite{brassard2000amplitude}. Measuring the spatial register of an amplified state $\ket{\phi^{N_t}}$ in the computational basis instead samples grid indices $j\in[N]$ from $p_j=|\phi_j^{N_t}|^2/\|\bb\phi^{N_t}\|_2^2$. Histograms and centroids computed from computational-basis shots refer to this squared-amplitude distribution, whereas concentration-weighted linear moments use the overlap protocol~\eqref{eq:linearObservableOverlap}. %[L06]

\subsection{Complete algorithm, error budget, and resource bounds}\label{subsec:growthAccuracy}

The amplification and precision budgets are formulated in terms of a certified upper bound $\widehat K_T\geq K_T$. For either output mode, Algorithm~\ref{alg:timeMarching} uses $\mathcal O(\widehat K_T/\sqrt{p_{\mathrm{out}}})$ calls to the pre-amplification circuit and its inverse, corresponding to $\mathcal O(N_t\widehat K_T/\sqrt{p_{\mathrm{out}}})$ single-step calls. The same certified bound fixes the per-step block error at $\eta=\mathcal O(\varepsilon\sqrt{p_{\mathrm{out}}}/(N_t\widehat K_T))$ and the arithmetic precision at $b=\mathcal O(\log(N_t\widehat K_T/(\varepsilon\sqrt{p_{\mathrm{out}}})))$. Theorem~\ref{thm:tmComplexity} states the constant-factor condition $\widehat K_T=\mathcal O(K_T)$; a looser supplied bound remains explicit in the amplification and precision costs. %[L07]

Mass conservation provides one problem-dependent way to certify such an upper bound. Combining it with the finite-time growth factors in~\eqref{eq:quantumGrowth} gives the following example. Establishing a constant-factor estimate requires additional information beyond mass conservation. %[L07]

\begin{remark}[A priori bounds for $\Lambda$ and $R_T$ via mass conservation]\label{rem:practicalBounds}
For nonzero conserved total mass $M_{\mathrm{tot}}=\sum_j\phi_j^0\neq0$, Cauchy--Schwarz and $\sum_i a_i=1$ give $|M_{\mathrm{tot}}|\leq\sqrt N\|\bb g^k\|$ (\ref{app:tmProof}), while velocity derivatives determine $\Lambda=\exp(\ell_{\mathrm{st}}\Delta t/2)$ independently in \ref{app:coefficientBounds}. Consequently, $\|\bb g^{N_t}\|_2\geq |M_{\mathrm{tot}}|/\sqrt{N}$ and $R_T\leq\sqrt{N}\|\bb g^0\|_2/|M_{\mathrm{tot}}|$. Under equilibrium initialization, this inequality provides an explicit $N_t$- and $D$-independent $\mathcal O(1)$ upper bound on $R_T$ whenever the initial scalar field has a ratio of the spatial root-mean-square value to the absolute spatial mean bounded uniformly under grid refinement, as holds for profiles with a grid-independent positive background and a uniformly bounded spatial root-mean-square value. For zero-mean oscillatory fields ($M_{\mathrm{tot}}=0$), consider a Fourier mode initialized at equilibrium. Once kinetic transients have decayed and the mode is well resolved, the continuum decay rate motivates the heuristic estimate $\|\bb g^{N_t}\|_2\sim e^{-D|\bb\kappa|^2 T}\|\bb g^0\|_2$ under spatially uniform advection. This estimate illustrates how $R_T$ can grow exponentially with $D|\bb\kappa|^2 T$. Section~\ref{subsec:exp2Ablation} measures the actual discrete norm attenuation. Such physical attenuation also enters the costs of other partial differential equation solvers that prepare normalized states \cite{jin2022quantum,Lin2023timeODE}. %[L07]

The mass bound furnishes the certified estimate
\[
 \widehat K_{T,\mathrm{mass}}
 =\Gamma_T\frac{\sqrt{N}\|\bb g^0\|}{|M_{\mathrm{tot}}|}\geq K_T,
 \qquad
 \frac{\widehat K_{T,\mathrm{mass}}}{K_T}
 =\frac{\sqrt{N}\|\bb g^{N_t}\|}{|M_{\mathrm{tot}}|}.
\]
When $M_{\mathrm{tot}}\neq0$, mass conservation supplies a certified bound; the uniform relation $\widehat K_{T,\mathrm{mass}}=\mathcal O(K_T)$ across arbitrary parameter families remains an additional condition.
\end{remark} %[L07]

By \eqref{Pru0ut} and Lemma~\ref{lem:recoverF}, the pre-amplification success probabilities for preparing $\ket{g^{N_t}}$ and $\ket{f^{N_t}}$ are %[L07,L06]
\begin{equation}\label{eq:tmSuccess}
\begin{aligned}
 \mathbb{P}_{\mathrm{TM}}&=\frac{\|\bb g^{N_t}\|^2}{\Gamma_T^2\|\bb g^0\|^2}=K_T^{-2},
 \\
 \mathbb{P}_{\mathrm{TM},f}&=\mathbb{P}_{\mathrm{TM}}\mathbb{P}_{\mathrm{rec}}
 \\ &=\frac{\|\bb f^{N_t}\|^2}{\Gamma_T^2\|\bb g^0\|^2}\geq\frac{a_*}{K_T^2}.
\end{aligned}
\end{equation}
Algorithm~\ref{alg:timeMarching} assembles the complete coherent time-marching procedure for both distribution and scalar output modes, and Theorem~\ref{thm:tmComplexity} establishes its distribution-mode query complexity. %[L07]

\begin{algorithm*}[!t]
\caption{Coherent QLBM time marching with equilibrium energy scaling}
\label{alg:timeMarching}
\begin{algorithmic}[1]
\Require Step count $N_t\geq 1$, relaxation $\tau^*\geq 1/2$, given exact initial scaled-state preparation unitary $U_0$ ($\ket{g^0}=U_0\ket{0^{n_d+n_g}}$; for equilibrium initialization, $U_0=V^0(I_{2^{n_d}}\otimes U_{\phi^0})$), coefficient access / rotation schedules for $\{a_i(t_k,\bb x_j)\}_{k=0}^{N_t}$, periodic shift unitary $P$, analytic step bound $\Lambda\geq\max_k\rho^k$, upper bound $\widehat K_T\geq K_T=\Lambda^{N_t}\|\bb g^0\|/\|\bb g^{N_t}\|$, target accuracy $\varepsilon\in(0,1)$, output mode $\mathsf{mode}\in\{\mathsf{dist},\mathsf{scalar}\}$, and a certified conditional output-success lower bound $0<p_{\mathrm{out}}\leq 1$ (take $p_{\mathrm{out}}=a_*$ for distribution recovery; scalar output requires $p_{\mathrm{out}}\leq\|\bb\phi^{N_t}\|^2/\|\bb g^{N_t}\|^2$).
\Ensure Normalized $\varepsilon$-accurate terminal state $\ket{\widetilde f^{N_t}}$ (if $\mathsf{mode}=\mathsf{dist}$) or $\ket{\widetilde\phi^{N_t}}$ (if $\mathsf{mode}=\mathsf{scalar}$) flagged by the mode-dependent success subspace specified below.
\State Allocate counter register $c$ ($n_c=\lceil\log_2(N_t+1)\rceil$ qubits), step flags $(a_P,a_C)$ ($2$ qubits), optional recovery flag $a_{\mathrm{rec}}$ ($1$ qubit if $\mathsf{mode}=\mathsf{dist}$), direction register $d$ ($n_d=\lceil\log_2 q\rceil$ qubits), and spatial grid register $g$ ($n_g=\lceil\log_2 N\rceil$ qubits), initialized to $\ket{0}^{\otimes\mathrm{all}}$ at precision $b=\mathcal O(\log(N_t\widehat K_T/(\varepsilon\sqrt{p_{\mathrm{out}}})))$; use per-step block error $\eta=\mathcal O(\varepsilon\sqrt{p_{\mathrm{out}}}/(N_t\widehat K_T))$ and terminal-block error $\mathcal O(\varepsilon\sqrt{p_{\mathrm{out}}})$, with sufficiently small constants.
\State Define forward pre-amplification circuit $\mathcal W$:
\State \hspace{1.2em} Apply $U_0$ on $(d,g)$ to prepare $\ket{g^0}$.
\For{$j = 1$ \textbf{to} $N_t$}
 \State Apply $(1,1,0)$ collision block-encoding $U_C^{j-1}$ on $(a_C,d,g)$ using $\theta_\alpha=2\arcsin(\sqrt{1/(2\tau^*)})$ and $V^{j-1}$.
 \State Apply $(\Lambda,1,0)$ weighted-streaming block-encoding $U_P^{j-1}=(I_2\otimes P_{\mathrm{ext}})D_\Lambda^{j-1}$ on $(a_P,d,g)$ (for steady uniform $\bb u$ and $\Lambda=1$, replace the ratio multiplexor by the padding flag flips $F_{\mathrm{pad}}$).
 \State Update failure counter via $S$: apply $\mathsf{ADD}$ on $c$, then $\mathsf{ADD}^\dagger$ on $c$ controlled on $(a_P,a_C)=(0,0)$.
\EndFor
\If{$\mathsf{mode}=\mathsf{dist}$}
 \State Apply $(1,1,0)$ recovery rotation $U_{\mathrm{rec}}$ of $(\mathcal A^{N_t})^{1/2}$ on $(a_{\mathrm{rec}},d,g)$ (success subspace: $c=0,a_P=0,a_C=0,a_{\mathrm{rec}}=0$).
\Else
 \State Apply $(V^{N_t})^\dagger$ on $(d,g)$ (success subspace: $c=0,a_P=0,a_C=0,d=0^{n_d}$, leaving $\ket{\phi^{N_t}}$ on $g$).
\EndIf
\State Use fixed-point or randomized amplitude amplification based on the supplied success-probability lower bound, with $m=\mathcal O(\widehat K_T/\sqrt{p_{\mathrm{out}}})$ calls to $\mathcal W$ and its inverse. The reflections defining $\mathcal Q=-\mathcal W R_{\mathrm{init}}\mathcal W^\dagger R_{\mathrm{good}}$ \cite{brassard2000amplitude,Gilyen2019QSVD}, or their phase-shifted versions, boost the success-flag probability to $\Omega(1)$. Standard Grover iterations require a calibrated success angle to select the iteration count.
\end{algorithmic}
\end{algorithm*} %[L07]

\begin{theorem}\label{thm:tmComplexity}
 Assume \eqref{eq:positivity} with fixed $a_*>0$, $\tau^*\geq1/2$, periodic permutation streaming, and coherent oracle access as in Section~\ref{subsec:structuredBlockEncoding}. Let $N_t\geq1$, $\bb f^0,\bb f^{N_t}\neq\bb 0$, and suppose $U_0,U_0^\dagger$ (preparing $\ket{g^0}$) and a constant-factor upper bound $\widehat K_T\geq K_T$ are given. For any $\varepsilon\in(0,1)$, Algorithm~\ref{alg:timeMarching} in distribution mode prepares $\ket{\widetilde f^{N_t}}$ satisfying $\|\ket{\widetilde f^{N_t}}-\ket{f^{N_t}}\|\leq\varepsilon$ with constant success probability and a verifiable success flag using $\mathcal O(N_t K_T)$ calls to the single-step block-encod\-ings $U_B^k$ and their inverses, $\mathcal O(K_T)$ calls to $U_0,U_0^\dagger$, and $\mathcal O(K_T)$ recovery coefficient queries, with per-step block precision $\eta=\mathcal O(\varepsilon\sqrt{a_*}/(N_tK_T))$ and recovery precision $\eta_{\mathrm{rec}}=\mathcal O(\varepsilon\sqrt{a_*})$ (requiring $b=\mathcal O(\log(N_tK_T/(\varepsilon\sqrt{a_*})))$ arithmetic bits).
\end{theorem}
The proof in \ref{app:tmProof} combines telescoping step-encoding errors with normalized-state recovery and terminal amplitude amplification. Because $\|\bb f^k\|\leq\|\bb g^k\|\leq a_*^{-1/2}\|\bb f^k\|$, for fixed $a_*$ the evolution-oracle query complexity in Theorem~\ref{thm:tmComplexity} takes the physical-norm form %[L07]
\begin{equation}\label{eq:tmFinalCost}
 \mathcal O\left(N_t\Gamma_T\frac{\|\bb f^0\|}{\|\bb f^{N_t}\|}\right),
\end{equation}
with total elementary gate complexity $\mathcal O\bigl(K_T(N_t C_{\mathrm{step}}+C_0+C_{\mathrm{rec}})\bigr)$, where $C_{\mathrm{step}}$, $C_0$, and $C_{\mathrm{rec}}$ denote the gate costs of one complete step including the counter update, initial preparation, and terminal recovery. Table~\ref{tab:inputCoverage} summarizes how each algorithmic input enters these bounds alongside the scope of our circuit verification. Register widths and module gate counts for converting query bounds into circuit costs are compiled in Table~\ref{tab:moduleResources} (\ref{app:completeAlgorithm}). %[L07]

\begin{table*}[!htb]
\centering\small
\caption{Algorithmic inputs, how they enter the bounds, and the scope of the implemented verification.}
\label{tab:inputCoverage}
\setlength{\tabcolsep}{5pt}
\begin{tabular}{>{\raggedright\arraybackslash}p{2.6cm}>{\raggedright\arraybackslash}p{5.9cm}>{\raggedright\arraybackslash}p{5.5cm}}
\toprule
Input or requirement & Role in the algorithm & Implemented verification \\
\midrule
Coefficient access & Reversible $b$-bit values for the arithmetic-oracle query model, or explicit angle schedules with their compilation cost. & Direct multiplexor circuits; no arithmetic-oracle implementation. \\
Initial-state access & Given $U_0,U_0^\dagger$; equilibrium initialization factors through $U_{\phi^0}$ and $V^0$. General kinetic data require separate preparation. & Equilibrium-initialized circuits; non-equilibrium memory is tested by classical trajectories. \\
Bounds $\Lambda,\widehat K_T$ & Analytic derivative bounds certify $\Lambda$; mass or other problem information can certify $\widehat K_T$. Costs retain the supplied bound unless it is tight within a constant factor. & Operator ratios, norm attenuation and success identities are evaluated on the specified cases. \\
Output lower bound $p_{\mathrm{out}}$ & $a_*$ suffices for distribution recovery. Scalar preparation requires an additional positive projection-probability bound. & Scalar readout on the equilibrium-initialized tests; no small-$p_{\mathrm{out}}$ sweep. \\
Precision and amplification & Step and terminal tolerances follow Algorithm~\ref{alg:timeMarching}. Fixed-point or randomized amplification can use a certified lower bound. & Floating-point exact-angle simulation and standard Grover iterations with a calibrated success angle; no finite-bit or fixed-point-amplification test. \\
\bottomrule
\end{tabular}
\end{table*}

\section{History-state formulation and comparison with time marching}\label{sec:QLSALBM}

The history-state and time-marching routes are compared here for the same scaled discrete operator $B^k$ and the same terminal-distribution state-preparation task. Time marching applies $B^k$ sequentially, whereas the history-state formulation embeds the energy-scaled recurrence~\eqref{gtime} into a global space-time linear system over all time levels. Using the same operator and output task focuses the comparison on history-matrix conditioning, terminal extraction and amplification, and the requirements for coefficient access. %[L08,L09]

\subsection{History system, terminal padding, and conditioning}

Stacking $\bb G_0=[\bb g^0;\ldots;\bb g^{N_t}]$ rewrites~\eqref{gtime} as a lower-bidiagonal system $\mathcal L_0\bb G_0=\bb F_0$, whose explicit matrix and causal inverse are given in \ref{app:qlsaProof}. Because $\bb g^{N_t}$ occupies only one of $N_t+1$ time blocks, projecting $\ket{G_0}$ directly onto the final step incurs an $\mathcal O(N_t^{-1})$ probability dilution whenever time-level norms remain comparable. Using standard terminal padding \cite{berry2014high}, we append $N_t$ identity copy equations $\bb g^{k+1}-\bb g^k=\bb 0$ for $k=N_t,\ldots,2N_t-1$ by setting $B^k=I$ on $N_t\leq k<2N_t$. The resulting $(2N_t+1)$-block vector $\bb G=[\bb g^0;\cdots;\bb g^{N_t};\bb g^{N_t};\cdots;\bb g^{N_t}]$ satisfies %[L08]
\begin{equation}\label{AxbLBM}
 \mathcal L\bb G=\bb F,\qquad
 \mathcal L=I-\sum_{k=0}^{2N_t-1}\ket{k+1}\bra{k}\otimes B^{k},
 \qquad \bb F=\ket{0}\otimes\bb g^0.
\end{equation}

Each nonzero block in the finite nilpotent inverse $\mathcal L^{-1}=\sum_{r=0}^{2N_t}(I-\mathcal L)^r$ traverses at most $N_t$ physical operators $B^k$ before encountering identity copies (\ref{app:qlsaProof}), so each such block is bounded by the single cumulative factor $\Gamma_T$. %[L08]

\begin{lemma}\label{lem:Lbound}
 The padded energy-scaled history matrix $\mathcal L$ in \eqref{AxbLBM} satisfies
 \[
  \begin{aligned}\sigma_{\max}(\mathcal L)&=\|\mathcal L\|_2\leq 1+\Lambda,\\
  \sigma_{\min}(\mathcal L)&=\|\mathcal L^{-1}\|_2^{-1}\geq\frac{1}{(2N_t+1)\Gamma_T},\end{aligned}
 \]
 and hence $\kappa_2(\mathcal L)\leq(1+\Lambda)(2N_t+1)\Gamma_T$.
\end{lemma}
Lemma~\ref{lem:Lbound} bounds the scaled history condition number by $(1+\Lambda)(2N_t+1)\Gamma_T$ prior to terminal time extraction, which scales linearly as $\mathcal O(N_t)$ under fixed-horizon refinement whenever $\Lambda$ and $\Gamma_T$ remain uniformly bounded. %[L08]

\begin{remark}[Comparison with the unscaled history matrix]\label{rem:unscaledHistoryComparison}
Let $\mathcal L_f$ denote the history matrix built from the unscaled steps $M^k$ with the same $N_t$ terminal copies. The block-diagonal basis change $\mathcal D$ applies $(\mathcal A^k)^{1/2}$ on physical time levels and repeats $(\mathcal A^{N_t})^{1/2}$ across the copy window, giving $\mathcal L_f=\mathcal D\mathcal L\mathcal D^{-1}$. With $a^*=\max_{k,i,j}a_i(t_k,\bb x_j)\leq1$, the bound $\kappa_2(\mathcal D)\leq\sqrt{a^*/a_*}$ yields %[L08]
\begin{equation}\label{eq:historySimilarityBounds}
  \frac{a_*}{a^*}\,\kappa_2(\mathcal L)
  \leq\kappa_2(\mathcal L_f)
  \leq\frac{a^*}{a_*}\,\kappa_2(\mathcal L)
  \leq a_*^{-1}\kappa_2(\mathcal L).
 \end{equation}
For a fixed positivity margin, the similarity factor is independent of $N_t$ and gives two-sided bounds on the history condition number. The scaled representation also supplies a collision LCU with unit normalization. The actual ordering of the two condition numbers depends on the coefficients: Section~\ref{subsec:exp6History} reports lower values for the scaled matrices in the time-refinement tests and a reversal near the positivity boundary $a_*\to 0$, consistent with~\eqref{eq:historySimilarityBounds}.
\end{remark} %[L08,L15]

\subsection{Coefficient access and costs of preparing the terminal state}

Sequential time marching accesses coefficients according to a classical time schedule, whereas assembling the global operator $\mathcal L$ requires coherent access to coefficients in a quantum superposition over time indices. On $n_{\mathrm{hist}}=\lceil\log_2(2N_t+1)\rceil$ time qubits, the multiplexed step oracle $\mathsf{HAM\mbox{-}T}_B$ selects $U_B^k$ on physical steps $0\leq k<N_t$ and an $I/\Lambda$ block-encoding across the copy window. Composing $\mathsf{HAM\mbox{-}T}_B$ with a truncated time shift and a two-term LCU encodes $\mathcal L/(1+\Lambda)$ at $\mathcal O(1)$ coefficient queries per call, with unused time labels and Hermitian dilation handled in \ref{app:qlsaProof}. %[L09]

To extract the terminal state from $\ket G$, we project the time register onto the uniform copy-window superposition $\ket{s_f}=(N_t+1)^{-1/2}\sum_{k=N_t}^{2N_t}\ket k$ with probability %[L09,L08]
\begin{equation}\label{eq:qlsaSuccess}
 \mathbb{P}_{\mathrm{time}}=\frac{(N_t+1)\|\bb g^{N_t}\|^2}{\|\bb G\|^2}
 \geq\frac{1}{2K_T^2},
\end{equation}
so that composing copy-window projection with physical recovery $U_{\mathrm{rec}}$ yields combined success probability $\mathbb P_{\mathrm{QLSA},f}=(N_t+1)\|\bb f^{N_t}\|^2/\|\bb G\|^2\geq a_*/(2K_T^2)$. Terminal padding thus removes the explicit $\mathcal O(N_t^{-1})$ history-length dilution while retaining the combined factor $K_T^{-2}$. In the resulting complexity, $\Gamma_T$ enters once through the condition-number bound $\kappa_2(\mathcal L)=\mathcal O((1+\Lambda)N_t\Gamma_T)$ and once through terminal extraction and amplification because $\mathbb P_{\mathrm{time}}^{-1/2}=\mathcal O(K_T)=\mathcal O(\Gamma_T R_T)$. Combining these two stages (\ref{app:qlsaProof}) determines the overall linear-system query complexity for the same task of preparing the quantum state encoding the terminal distribution as in Theorem~\ref{thm:tmComplexity}. %[L09,L08]

\begin{theorem}\label{thm:qlsaComplexity}
 Under the assumptions of Theorem~\ref{thm:tmComplexity} (with $\mathsf{HAM\mbox{-}T}_B$ providing coherent access to coefficients in superposition over time indices), a block-encoding quantum linear systems algorithm (QLSA) with query complexity linear in $\kappa$ \cite{Costa2021QLSA} combined with rank-one copy-window projection and physical recovery prepares $\ket{\widetilde f^{N_t}}$ satisfying $\|\ket{\widetilde f^{N_t}}-\ket{f^{N_t}}\|\leq\varepsilon$ with constant success probability and a success flag using %[L09]
 \begin{equation}\label{eq:qlsaFinalCost}
  \mathcal O\left((1+\Lambda)N_t\Gamma_T K_T\log\frac{K_T}{\varepsilon}\right)
  =\mathcal O\left((1+\Lambda)N_t\Gamma_T^2\frac{\|\bb f^0\|}{\|\bb f^{N_t}\|}\log\frac{K_T}{\varepsilon}\right)
 \end{equation}
 total queries to the block-encoding of $\mathcal L$ (each costing $\mathcal O(1)$ $\mathsf{HAM\mbox{-}T}_B$ calls) and to the right-hand-side preparation oracle $\ket 0\ket{g^0}$, plus $\mathcal O(K_T)$ terminal recovery coefficient queries, at history block precision $\mathcal O\bigl(\varepsilon\sqrt{a_*}/((1+\Lambda)N_t\Gamma_T K_T)\bigr)$ (requiring logarithmic bit precision).
\end{theorem}
Table~\ref{tab:tmVsQlsa} first compares the query upper bounds under the two oracle models, absorbing $(1+\Lambda)$ into $\mathcal O(N_t\Gamma_T)$ whenever $\Lambda=\mathcal O(1)$. For spatially uniform steady advection ($\Lambda=\Gamma_T=1$), both upper bounds scale nearly linearly with $N_t$ up to $R_T$ and logarithmic factors. Under variable velocity, the history-state bound contains one factor of $\Gamma_T$ from conditioning and a second through terminal extraction and amplification. Coherent time marching calls $U_0$ $\mathcal O(K_T)$ times, with no additional explicit $N_t$ factor beyond the possible dependence already contained in $K_T=\Lambda^{N_t}R_T$. As in Theorem~\ref{thm:tmComplexity}, writing both bounds in terms of $K_T$ uses a constant-factor estimate $\widehat K_T=\mathcal O(K_T)$; any looser certified bound remains explicit in the amplification and precision budgets. %[L09,L08]

The two methods also differ in coefficient access, a distinction separate from the query counts. Time marching executes $U_B^k$ under the classical schedule $k=0,\ldots,N_t-1$, whereas the history-state formulation requires $\mathsf{HAM\mbox{-}T}_B$ to access coefficients coherently in a superposition of time indices. A comparison of compiled gate counts or hardware runtimes must also account for direct multiplexed rotations, reversible arithmetic, quantum linear-system subroutines, and state preparation under the respective access models. %[L09]

\begin{table*}[!htb]
\centering\small
\caption{Comparison of the two energy-scaled QLBM state-preparation methods for fixed positivity margin $a_*>0$ ($R_T=\|\bb g^0\|/\|\bb g^{N_t}\|\simeq\|\bb f^0\|/\|\bb f^{N_t}\|$, $K_T=\Gamma_T R_T$). The norm ratios are equivalent up to constant factors for the fixed positivity margin.}
\label{tab:tmVsQlsa}
\begin{tabular}{>{\raggedright\arraybackslash}p{3.5cm}>{\raggedright\arraybackslash}p{5.5cm}>{\raggedright\arraybackslash}p{5.8cm}}
\toprule
Comparison metric & Coherent Time Marching (Sec.~\ref{sec:ADEquantum}) & Padded History-State QLSA (Sec.~\ref{sec:QLSALBM}) \\
\midrule
Evolution oracle queries & $\mathcal O(N_t\Gamma_T R_T)$ calls to $U_B^k,(U_B^k)^\dagger$ & $\mathcal O\bigl((1+\Lambda)N_t\Gamma_T^2 R_T\log(K_T/\varepsilon)\bigr)$ calls to $\mathsf{HAM\mbox{-}T}_B$ \\
Initial-state $U_0$ queries & $\mathcal O(\Gamma_T R_T)$ (no additional explicit $N_t$ factor; $K_T=\Gamma_T R_T$ may depend on $N_t$) & $\mathcal O\bigl((1+\Lambda)N_t\Gamma_T^2 R_T\log(K_T/\varepsilon)\bigr)$ \\
Logical ancilla overhead (excluding arithmetic work registers) & $3+\lceil\log_2(N_t+1)\rceil$ (counter $+$ $2$ step flags $+$ $1$ recovery flag) & $\lceil\log_2(2N_t+1)\rceil+\mathcal O(1)$ (padded time $+$ QLSA flags) \\
Time-coefficient access & Classical step-by-step schedule $k=0,1,\ldots,N_t-1$ & Coherent quantum superposition over $k\in\{0,\ldots,2N_t\}$ \\
Origin of $\Gamma_T$ scaling & Enters through terminal extraction and amplification via $1/K_T$ & Enters through conditioning, $\kappa_2(\mathcal L)=\mathcal O(N_t\Gamma_T)$, and terminal extraction and amplification, $\mathbb P_{\mathrm{time}}^{-1/2}=\mathcal O(K_T)$ \\
Empirical validation in Sec.~\ref{sec:numerics} & Classical normalization diagnostics ($N_t\leq 16{,}384$) $+$ full elementary-gate statevector simulation, Grover amplitude amplification \& Hadamard test & Exact singular-value decompositions of padded/unpadded history matrices, linear solves \& positivity-margin diagnostics \\
\bottomrule
\end{tabular}
\end{table*}

\section{Numerical and quantum-circuit verification}\label{sec:numerics}

Verification follows the logical dependencies of the framework: the closed kinetic recurrence must preserve the target transport dynamics, equilibrium energy scaling must suppress artificial subnormalization without distorting the discrete trajectory, and the compiled quantum modules must reproduce the encoded output. Accordingly, the six studies in Sections~\ref{subsec:exp1Gaussian}--\ref{subsec:exp6History} combine classical double-precision LBM trajectories, noiseless elementary-gate statevector simulations, and exact diagnostics based on assembled history matrices. The statevector simulations evaluate circuit synthesis and readout on small registers, while the history-matrix diagnostics evaluate conditioning and padding by classical linear algebra. Throughout, we separate continuum discretization error $E_{\mathrm{ADE}}^k$, discrete representation error $E_{\mathrm{eq}}^k$, and mass drift $E_{\mathrm{mass}}^k$: %[L10,L11,L12,L13,L14,L15]
\begin{equation}\label{eq:errorMetrics}
\begin{aligned}
 E_{\mathrm{ADE}}^k&=\frac{\|\bb\phi_{\mathrm{LBM}}^k-\bb\phi_{\mathrm{ADE}}^k\|_2}{\|\bb\phi_{\mathrm{ADE}}^k\|_2},
 \\
 E_{\mathrm{eq}}^k&=\frac{\|\bb f_{\mathrm{LBM}}^k-\bb f_{\mathrm{rec}}^k\|_2}{\|\bb f_{\mathrm{LBM}}^k\|_2},
 \\
 E_{\mathrm{mass}}^k&=\frac{\bigl|\sum_j\phi_{\mathrm{LBM},j}^k-\sum_j\phi_j^0\bigr|}{\bigl|\sum_j\phi_j^0\bigr|},
\end{aligned}
\end{equation}
where $\bb f_{\mathrm{LBM}}^k$ evolves directly via \eqref{ftime} and $\bb f_{\mathrm{rec}}^k=(\mathcal A^k)^{1/2}\bb g^k$ is recovered from the energy-scaled trajectory $\bb g^{k+1}=\widetilde P^k\widetilde C^k\bb g^k$. The relative mass drift $E_{\mathrm{mass}}^k$ applies only when the initial total scalar mass is nonzero. The zero-mean dissipation test instead reports the norm ratio $R_T$. Each test specifies its simulation parameters, and \ref{app:benchmarkDetails} records the Gaussian reference solution and steady-test settings. %[L10,L13]

\subsection{Discrete equivalence, steady transport, and non-equilibrium memory}\label{subsec:exp1Gaussian}\label{subsec:exp3Memory}

The first test asks whether equilibrium energy scaling preserves the closed discrete kinetic trajectory under uniform steady transport, for which $\Lambda=\Gamma_T=1$. On 1D D1Q3 ($N_x=128$) and 2D D2Q5 ($64\times 64$) lattices in units $\Delta x=\Delta t=1$, we evolve equilibrium-initialized periodic Gaussians for $\tau^*\in\{0.8,1.0,1.3\}$ and also compare the scalar fields with the analytical reference in \ref{app:benchmarkDetails}. Across all twelve configurations (Fig.~\ref{fig:gaussianCompact} and Table~\ref{tab:gaussianCombined}), the distribution recovered from the energy-scaled trajectory matches classical LBM to floating-point roundoff ($E_{\mathrm{eq}}^k\leq 5.78\times 10^{-15}$ and $E_{\mathrm{mass}}^k\leq 4.41\times 10^{-15}$). %[L10]

The analytical ADE comparison answers the separate question of continuum accuracy: the finite-grid deviation $E_{\mathrm{ADE}}^k\in[7.28\times 10^{-4},4.81\times 10^{-2}]$ reflects the classical finite-Mach advection error qualitatively indicated by Lemma~\ref{lem:consistency}. The fixed-$c$ expansion as $\Delta t\to 0$ provides a qualitative interpretation of the error observed at $\Delta x=\Delta t=1$. Convergence under diffusive refinement is measured separately in Section~\ref{subsec:exp4Unsteady}. %[L10]

For the encoded output on the same uniform steady trajectories, $\Lambda=\Gamma_T=1$ gives a success probability before amplitude amplification of $\mathbb P_{\mathrm{TM},f}\in[0.500,0.516]$ in 1D and $[0.171,0.246]$ in 2D for preparing the physical distribution state. Conditional on the recovered terminal $\bb g$ state, direct macroscopic scalar projection $(V^{N_t})^\dagger$ achieves the conditional probability $\mathbb P_\phi>0.983$ in every case without amplitude amplification. %[L10]

\begin{figure*}[!htb]
 \centering
 \includegraphics[width=0.92\linewidth]{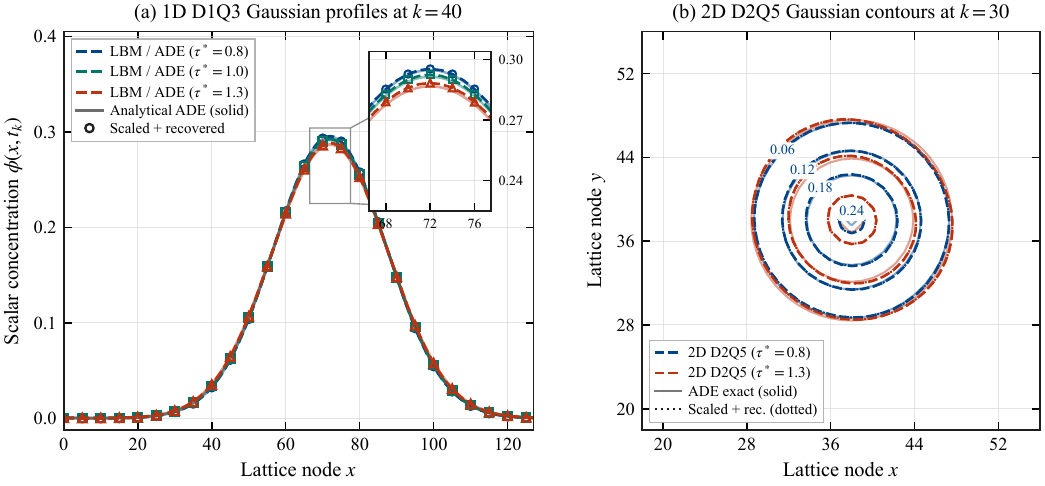}
 \caption{Steady periodic Gaussian transport in lattice units ($\Delta x=\Delta t=1$, $\Lambda=\Gamma_T=1$). (a) 1D D1Q3 scalar profiles at $k=40$ ($N_x=128$, $u=0.2$) for $\tau^*\in\{0.8,1.0,1.3\}$. (b) 2D D2Q5 scalar contours at $k=30$ ($64\times 64$, $\bb u=(0.2,0.2)$, levels $0.06,0.12,0.18,0.24$) for $\tau^*\in\{0.8,1.3\}$. Classical LBM and energy-scaled recovered profiles coincide to floating-point roundoff.}
 \label{fig:gaussianCompact}
\end{figure*}

The second test asks whether the same macroscopic initial scalar field determines the subsequent scalar trajectory when the kinetic state is not fully relaxed \cite{wawrzyniak2025linearized,nagel2026multiple,matyka2021memory}. We therefore compare two 1D D1Q3 initial states ($N_x=64$, $u=0.15$, $\Delta x=\Delta t=1$) sharing $\phi^0(x)=0.5+0.25\exp(-(x-32)^2/(2\cdot 6^2))$: one starts at local equilibrium $\bb f^{0,\mathrm{eq}}=\mathcal A^0 E_I^\top\bb\phi^0$, whereas the other adds the zero-mass non-equilibrium flux perturbation $(-2\psi,3\psi,-\psi)^\top$ with $\psi(x)=0.08\sin(2\pi x/N_x)$. Thus $E_I\bb f^0=\bb\phi^0$ holds identically at $k=0$, while the two initial distributions differ in their non-equilibrium flux content. %[L11]

Figure~\ref{fig:noneqMemory} tracks the relative scalar trajectory separation $\|\bb\phi_{\mathrm{neq}}^k-\bb\phi_{\mathrm{eq}}^k\|_2/\|\bb\phi_{\mathrm{eq}}^k\|_2$ across $\tau^*\in\{0.6,0.8,1.0,1.5\}$. At $\tau^*=1.0$ ($\alpha=0$), collision projects $\bb f^0$ onto local equilibrium in one step and reduces the flux perturbation to roundoff ($1.32\times 10^{-16}$ at $k=1$ and $6.33\times 10^{-16}$ at $k=10$). For every tested $\tau^*\neq 1$, the retained factor $\alpha=1-1/\tau^*\neq 0$ instead converts the initial kinetic flux divergence into macroscopic scalar differences of $2.62\%$ ($\tau^*=0.6$), $0.98\%$ ($\tau^*=0.8$), and $1.31\%$ ($\tau^*=1.5$) after one step, which persist at $1.55\%$, $0.78\%$, and $1.91\%$ at $k=10$. This contrast shows that the macroscopic initial field alone does not retain the non-equilibrium flux memory away from full relaxation; propagating the full distribution is therefore necessary for the tested trajectories, and the energy-scaled evolution reproduces those trajectories over all $30$ steps with $\max_k E_{\mathrm{eq}}^k=4.29\times 10^{-15}$. %[L11]

\begin{figure*}[!htb]
 \centering
 \includegraphics[width=0.91\linewidth]{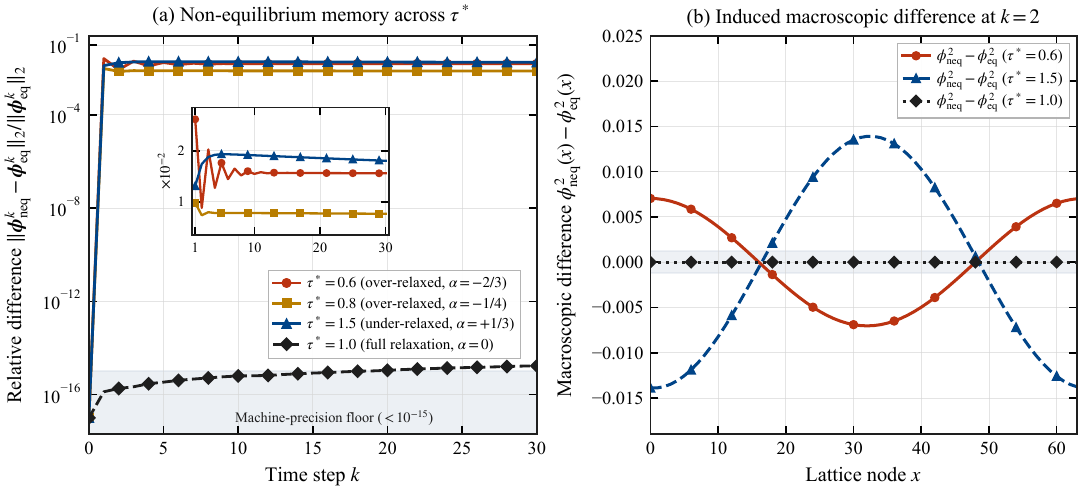}
 \caption{Non-equilibrium kinetic memory test on 1D D1Q3 ($N_x=64$, $u=0.15$). (a) Relative macroscopic scalar separation between equilibrium and flux-perturbed initial distributions sharing identical $\bb\phi^0$ with one-step memory erasure at $\tau^*=1.0$ versus persistent non-equilibrium memory for $\tau^*\in\{0.6,0.8,1.5\}$. (b) Scalar difference profiles $\phi_{\mathrm{neq}}^k-\phi_{\mathrm{eq}}^k$ at $k=2$ recovered from the energy-scaled perturbed evolution.}
 \label{fig:noneqMemory}
\end{figure*}

\subsection{Subnormalization ablation, physical dissipation, and unsteady refinement}\label{subsec:exp2Ablation}\label{subsec:exp4Unsteady}

The next tests distinguish normalization loss caused by the representation from physical diffusive attenuation, then examine energy scaling for unsteady flows under grid refinement. For the representation ablation, Fig.~\ref{fig:ablationRegimes}(a)--(b) compares unscaled and energy-scaled sequential normalizations for the same $50$-step evolution on a $6\times 6$ D2Q5 lattice at $\tau^*=0.8$ under the steady nonuniform vortex field $u_x=0.10+0.04\sin(2\pi x)\cos(2\pi y)$ and $u_y=0.08-0.04\cos(2\pi x)\sin(2\pi y)$ ($\Delta x=1/6$, $\Delta t=\Delta x^2$, $m=1,\ldots,50$), with the ordered propagator norms shown as a common reference. %[L12]

Although the ordered $50$-step propagators retain $\mathcal O(1)$ spectral norms ($\|(M^0)^{50}\|_2=1.0547$ and $\|(B^0)^{50}\|_2=1.00003$), non-normal collision gives the single-step unscaled operator $\|M^0\|_2=1.0891>1$, so any sequential block-encoding with $\alpha_{\mathrm{unsc}}\geq\|M^0\|_2$ accumulates subnormalization $\|M^0\|_2^{50}=71.26$. Equilibrium energy scaling lowers the single-step bound to $\rho^0=1.00924$ ($\|B^0\|_2=1.00011$) and the cumulative $50$-step subnormalization to $(\rho^0)^{50}=1.5841$. For this same $6\times 6$, $\tau^*=0.8$ trajectory, the terminal distribution-state success probability after $50$ steps rises from $\mathbb P_{\mathrm{unsc}}=1.95\times 10^{-4}$ to $\mathbb P_{\mathrm{TM},f}=8.76\times 10^{-2}$, an approximately $450\times$ terminal-probability gain that corresponds to more than a $21\times$ reduction in the ideal amplitude-amplification call factor. For this evolution, these two quantities refer specifically to terminal-state probability and the ideal amplitude-amplification call factor for the two sequential block-encodings; comparing gate counts or runtimes across the distinct protocols in Table~\ref{tab:relatedConstructions} requires separate resource accounting. %[L12]

The separate physical-dissipation test shows that energy scaling preserves the attenuation of the encoded field described in Remark~\ref{rem:practicalBounds} (Fig.~\ref{fig:ablationRegimes}(c)). On a 1D D1Q3 grid ($N_x=64$, $\tau^*=1.3$, $u=0.1$) over $160$ steps, the scaled-solution norm ratio $R_T=\|\bb g^0\|/\|\bb g^k\|$ for the positive-mean profile $\phi_{\mathrm{pos}}^0(x)=1+0.5\sin(2\pi x/L_x)$ reaches $1.0318$ at $k=160$ and remains strictly below the a priori mass-conservation ceiling $\sqrt{N_x}\|\bb g^0\|_2/|M_{\mathrm{tot}}|=1.0607$. For the zero-mean mode $\phi_{\mathrm{zero}}^0(x)=\sin(2\pi x/L_x)$ ($M_{\mathrm{tot}}=0$), diffusive decay instead drives $R_T$ exponentially upward to $1.4838$ at $k=160$. An independent scan over $u/c\in[0.02,0.328]$ ($a_*\in[2.67\times 10^{-3},0.157]$) confirms $\mathbb P_{\mathrm{rec}}\in[0.500,0.554]>a_*$ in agreement with Lemma~\ref{lem:recoverF}, thereby separating recovery probability from the scaled-solution norm ratio $R_T$. %[L06,L07]

\begin{figure*}[!htb]
 \centering
 \includegraphics[width=0.99\linewidth]{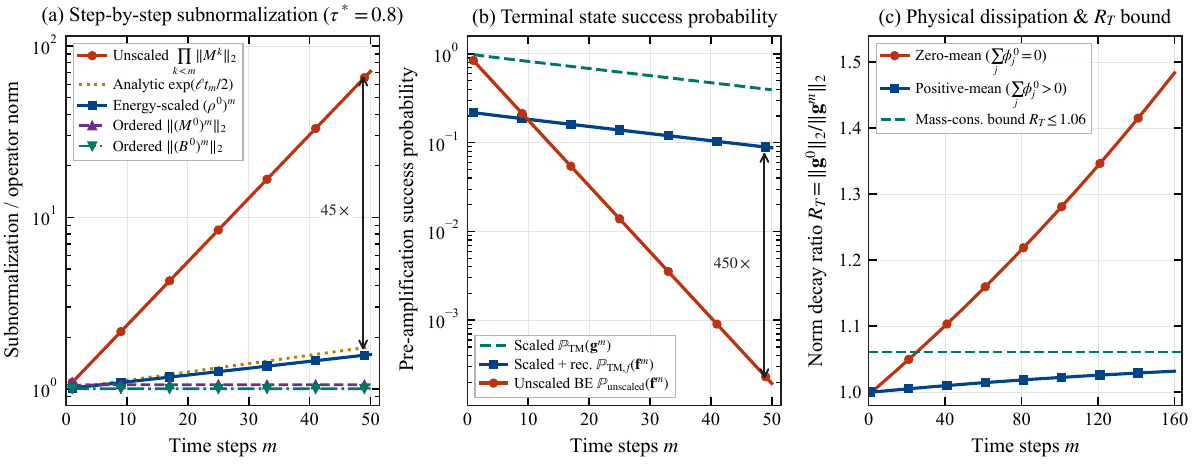}
 \caption{Effects of scaling on sequential normalization loss and physical diffusive attenuation. (a) Cumulative step-by-step subnormalization $\|M^0\|_2^m$ versus $(\rho^0)^m$ alongside ordered powers $\|(M^0)^m\|_2$ and $\|(B^0)^m\|_2$ on $6\times 6$ D2Q5 ($\tau^*=0.8$), where energy scaling lowers 50-step subnormalization from $71.26$ to $1.5841$. (b) Corresponding terminal distribution-state success probabilities from $1.95\times 10^{-4}$ to $8.76\times 10^{-2}$ ($\approx 450\times$ gain at $m=50$). (c) A priori mass-conservation ceiling on $R_T$ for positive-mean data versus growth of $R_T$ caused by diffusive decay of a zero-mean mode.}
 \label{fig:ablationRegimes}
\end{figure*}

The oscillatory spatially uniform velocity test addresses continuum accuracy under time-dependent scaling $(\mathcal A^{k+1})^{-1/2}P(\mathcal A^k)^{1/2}$. On $[0,1)^d$ up to $T=0.125$ at $\tau^*=0.8$ under diffusive scaling $\Delta x=1/n$, $\Delta t=\Delta x^2$ ($D=0.1$), we prescribe $u(t)=0.8+0.6\sin(8\pi t)$ in 1D D1Q3 ($n=16,32,64,128$, $N_t=32\ldots 2048$, $\sigma_0=0.12$) and $(u_x(t),u_y(t))=(0.8+0.5\sin(8\pi t),\allowbreak\,0.5+0.4\cos(8\pi t))$ in 2D D2Q5 ($n=16,32,64$, $N_t=32\ldots 512$, $\sigma_0=0.12$), and compare with the exact analytical Gaussian~\eqref{eq:periodicGaussian} using the integrated displacement $\bb s(t)=\int_0^t\bb u(\xi)\,\d\xi$. Although $\nabla\bb u=\bb 0$, temporal acceleration $\partial_t\bb u\neq\bb 0$ produces $\rho^k>1$ at finite lattice speed $c$. Figure~\ref{fig:unsteadyRefinement}(b) and Table~\ref{tab:unsteadyConvergence} show empirical second-order spatial convergence to the unsteady continuum ADE ($\mathcal O(\Delta x^2)=\mathcal O(\Delta t)$, orders $1.98$--$2.00$) together with near-roundoff recovery of the discrete scaled trajectory ($E_{\mathrm{eq}}^{N_t}\leq 2.56\times 10^{-13}$ after $N_t=2048$ steps). As diffusive refinement increases $c=n$, the temporal term $(3/c)\partial_t\bb u$ diminishes and the finest tested grids give $\Gamma_{\mathrm{emp}}=\prod_{k=0}^{N_t-1}\rho^k=1.014$ in 1D and $1.030$ in 2D. %[L13]

The nonuniform space-time vortex test instead diagnoses cumulative normalization, without a separate continuum-error measurement, on $[0,1)^2$ up to $T=1$ ($\Delta x=1/n$, $\Delta t=\Delta x^2$, $N_t=n^2$, $\tau^*=1.3$) under the field %[L13]
 \[
 \begin{gathered}
  u_x(x,y,t)=\bigl(0.10+0.04\sin(2\pi x)\cos(2\pi y)\bigr)\cos(2\pi t),
  \\
  u_y(x,y,t)=\bigl(0.08-0.04\cos(2\pi x)\sin(2\pi y)\bigr)\cos(2\pi t).
  \end{gathered}
 \]
Table~\ref{tab:matrixValidation} verifies on $n=4,8$ that $\|M^k\|_2>1.034$ in all cases, whereas $\|B^k\|_2\leq\rho^k$ and $\|B^k-\widetilde P^k\widetilde C^k\|_F/\|B^k\|_F$ remains at roundoff level. %[L02,L13]

Across $n\in\{4,8,16,32,64,128\}$ ($c=n\geq 4$), Lemma~\ref{lem:spaceTimeGrowth} gives the single grid-independent space-time derivative bound $\ell_{\mathrm{st}}=3(0.08\pi+0.28\pi/4)/(1-3(0.14)/4)\approx 1.5796$ ($\Gamma_T\leq e^{\ell_{\mathrm{st}}/2}\approx 2.2029$); for the steady counterpart, it gives $\ell\approx 0.8424$ ($e^{\ell/2}\approx 1.5238$). The recorded diagnostics show that both the maximum link growth rate $\max_k\log(\rho^k)/\Delta t\in[0.3778,0.3842]\leq\ell_{\mathrm{st}}/2\approx 0.7898$ and the exact $N_t$-step product $\prod_{k=0}^{N_t-1}\rho^k\in[1.2772,1.3940]\leq 2.2029$ remain uniformly bounded as the number of classical diagnostic steps grows from $16$ to $16{,}384$. Figure~\ref{fig:unsteadyRefinement}(c) verifies the stated normalization bound over the tested refinement sequence, while Lemma~\ref{lem:spaceTimeGrowth} supplies the analytic grid-independent result. %[L13]

\begin{figure*}[!htb]
 \centering
 \includegraphics[width=0.96\linewidth]{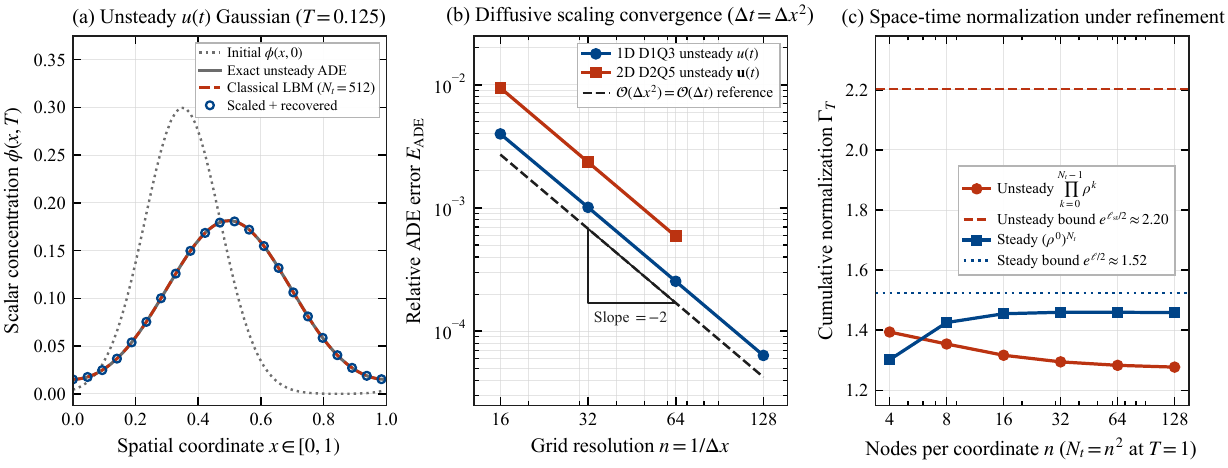}
 \caption{Unsteady velocity verification and diffusive mesh refinement. (a) 1D D1Q3 Gaussian profiles at $T=0.125$ on $n=64$ nodes ($N_t=512$). (b) Empirical second-order diffusive convergence ($\Delta t=\Delta x^2$) of 1D D1Q3 and 2D D2Q5 under oscillatory flow $\bb u(t)$ against the exact analytical solution~\eqref{eq:periodicGaussian}. (c) Cumulative normalization $\prod_{k=0}^{N_t-1}\rho^k$ at $T=1$ across $n=4\ldots 128$ ($N_t=16\ldots 16{,}384$) for 2D steady and unsteady vortex flows compared against Lemma~\ref{lem:spaceTimeGrowth}.}
 \label{fig:unsteadyRefinement}
\end{figure*}

\begin{table*}[!htb]
\centering\small
\caption{Diffusive refinement ($\Delta x=1/n$, $\Delta t=\Delta x^2$, $T=0.125$, $\tau^*=0.8$) under oscillatory unsteady advection $\bb u(t)$ reporting continuum ADE convergence order, scaled equivalence error $E_{\mathrm{eq}}^{N_t}$, trajectory normalization $\Gamma_{\mathrm{emp}}=\prod_{k=0}^{N_t-1}\rho^k$, and terminal success probability $\mathbb P_{\mathrm{TM},f}$.}
\label{tab:unsteadyConvergence}
\begin{tabular}{lrrrrrrr}
\toprule
Model & $n$ & $N_t$ & $E_{\mathrm{ADE}}^{N_t}$ & Order & $E_{\mathrm{eq}}^{N_t}$ & $\Gamma_{\mathrm{emp}}$ & $\mathbb P_{\mathrm{TM},f}$ \\
\midrule
1D D1Q3 & 16  & 32   & 3.99e-03 & ---  & 4.05e-15 & 1.1249 & 0.2418 \\
        & 32  & 128  & 1.01e-03 & 1.98 & 1.55e-14 & 1.0585 & 0.2713 \\
        & 64  & 512  & 2.54e-04 & 1.99 & 6.22e-14 & 1.0286 & 0.2868 \\
        & 128 & 2048 & 6.35e-05 & 2.00 & 2.56e-13 & 1.0142 & 0.2950 \\
\midrule
2D D2Q5 & 16  & 32   & 9.43e-03 & ---  & 2.69e-16 & 1.1274 & 0.0654 \\
        & 32  & 128  & 2.36e-03 & 2.00 & 4.01e-16 & 1.0613 & 0.0729 \\
        & 64  & 512  & 5.89e-04 & 2.00 & 4.19e-15 & 1.0303 & 0.0771 \\
\bottomrule
\end{tabular}
\end{table*}

\subsection{Gate-level implementation, amplification, and observable readout}\label{subsec:exp5Quantum}

The gate-level study proceeds from a single compiled module to multi-step propagation, amplification and readout, and finally grid-size scaling. Every module ($U_0$, $V^k$, $U_C^k$, $U_P^k$, $S$, $U_{\mathrm{rec}}$, Grover reflections, and Hadamard-test control) is compiled explicitly as a sequence of \texttt{\{X, H, Z, RY, RZ, P, CNOT\}} gates. The tested circuits use exact rotation angles and the multiplexed rotations compiled directly according to \ref{app:circuitSynthesis}, with zero scratch ancillas. Noiseless sequential statevector simulation executes these circuits; the reported gate counts apply to this compilation model. %[L14]

\begin{figure*}[!htb]
 \centering
 \includegraphics[width=0.94\linewidth]{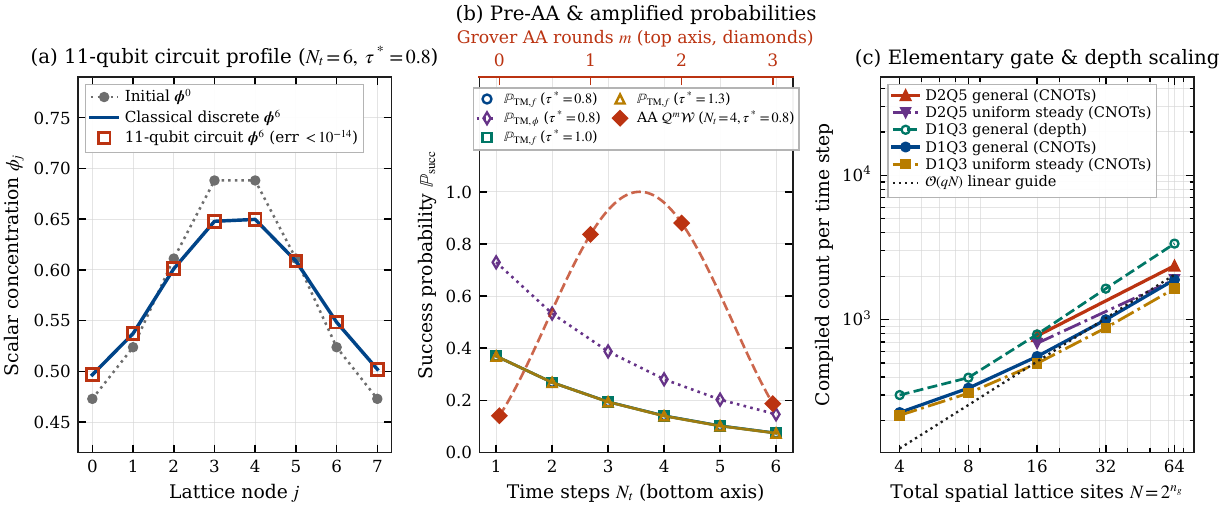}
 \caption{Elementary-gate quantum circuit verification over \texttt{\{X, H, Z, RY, RZ, P, CNOT\}}. (a) Terminal spatial scalar amplitudes from 11-qubit coherent time marching ($N=8$, $N_t=6$, unsteady nonuniform flow) compared against classical LBM at representative $\tau^*=0.8$ (with all $\tau^*\in\{0.8,1.0,1.3\}$ summarized in panel~(b) and the text). (b) Pre-amplification probabilities $\mathbb P_{\mathrm{TM},f}$ and $\mathbb P_{\mathrm{TM},\phi}$ across $\tau^*\in\{0.8,1.0,1.3\}$ versus $N_t$ (bottom axis) alongside Grover-amplified probabilities for $\mathcal Q^m\mathcal W$ at $N_t=4$, $\tau^*=0.8$ (top axis). (c) Compiled single-step CNOT count and circuit depth across $N=4\ldots 64$ for valid inputs with general coefficients or with uniform steady coefficients, with zero-padding of unused directions retained in the uniform branch.}
 \label{fig:quantumVerification}
\end{figure*}

\begin{table*}[!htb]
\centering\small
\caption{Compiled elementary-gate (\texttt{\{1Q, CNOT\}}) audit for single-step block-encodings $U_B^0$ ($\tau^*=0.8$, unsteady nonuniform flow) and 11-qubit Grover-amplified $N_t=4$ time marching on D1Q3 ($N=8$).}
\label{tab:quantumCircuitAudit}
\resizebox{\linewidth}{!}{%
\begin{tabular}{lrrrrrrr}
\toprule
Compiled circuit & Qubits & 1Q gates & CNOTs & Total & Depth & Empirical prob.\ / defect & Theory / State err. \\
\midrule
D1Q3 single step $U_B^0$ ($N=8$)     &  7 &   240 &   202 &    442 &  372 & $\|U^\dagger U-I\|_2=9.34\text{e-}15$ & BE err.\ $2.77\text{e-}15$ \\
D2Q5 single step $U_B^0$ ($4\times 4$) &  9 &   476 &   638 &  1,114 &  980 & $\|U^\dagger U-I\|_2=1.97\text{e-}14$ & BE err.\ $4.63\text{e-}15$ \\
\midrule
D1Q3 $N_t=4$ AA ($m=0$, $\mathcal W$)         & 11 & 1,303 & 1,101 &  2,404 & 1,565 & $\mathbb P_{\mathrm{succ}}=0.14097424$ & $0.14097424$ ($1.87\text{e-}15$) \\
D1Q3 $N_t=4$ AA ($m=1$, $\mathcal Q\mathcal W$)   & 11 & 6,093 & 5,375 & 11,468 & ---   & $\mathbb P_{\mathrm{succ}}=0.83662545$ & $0.83662545$ ($4.28\text{e-}15$) \\
D1Q3 $N_t=4$ AA ($m=2$, $\mathcal Q^2\mathcal W$) & 11 & 10,883 & 9,649 & 20,532 & ---  & $\mathbb P_{\mathrm{succ}}=0.88002846$ & $0.88002846$ ($4.33\text{e-}15$) \\
D1Q3 $N_t=4$ AA ($m=3$, $\mathcal Q^3\mathcal W$) & 11 & 15,673 & 13,923 & 29,596 & --- & $\mathbb P_{\mathrm{succ}}=0.18708525$ & $0.18708525$ ($1.01\text{e-}14$) \\
\bottomrule
\end{tabular}}
\end{table*}

At the single-module level, extracting the full $2^{n_q}\times 2^{n_q}$ circuit unitaries of $U_B^0$ under unsteady nonuniform advection ($\tau^*=0.8$) on 1D D1Q3 ($N=8$, $7$ qubits) and 2D D2Q5 ($4\times 4$, $9$ qubits) confirms global unitarity $\|U_B^{0\dagger}U_B^0-I\|_2\leq 1.97\times 10^{-14}$ and physical block-encoding defect $\|\Lambda(\bra{00}_{a_P,a_C}\otimes I)U_B^0(\ket{00}_{a_P,a_C}\otimes I)-B^0\|_2\leq 4.63\times 10^{-15}$ (Table~\ref{tab:quantumCircuitAudit}). These results verify the exact-angle elementary-gate realization of the intended successful block for the tested $7$- and $9$-qubit circuits. %[L14]

\ifdefined\CPCReviewLayout\begingroup\clubpenalty=10000\widowpenalty=10000\fi
The multi-step test uses 1D D1Q3 ($N=8$) under the unsteady nonuniform flow $u(x,t)=(0.12+0.05\sin(2\pi x))\allowbreak\cos(2\pi t/12)$ and simulates $N_t=1,\ldots,6$ coherent steps across $\tau^*\in\{0.8,1.0,1.3\}$ on $11$ logical qubits ($n_c=3$ counter, $a_{\mathrm{rec}},a_P,a_C$ flags, $n_d=2$ direction, and $n_g=3$ grid). Across all $18$ runs (Fig.~\ref{fig:quantumVerification}(a)--(b)), the postselected amplitudes rescaled by $\Gamma_T\|\bb g^0\|$ reproduce the unnormalized physical vectors $\bb f^{N_t}$ and $\bb\phi^{N_t}$ with relative Euclidean errors at most $9.34\times 10^{-15}$ and $8.95\times 10^{-15}$. These quantities are relative field-vector errors rather than normalized-state distances, and the empirical flag probabilities match~\eqref{eq:tmSuccess} within a maximum relative discrepancy of $1.79\times 10^{-14}$. A direct $N_t=3$ comparison between zero-counter-controlled steps $\mathcal V_j=S\overline U_j$ and uncontrolled steps $\widetilde{\mathcal V}_j=S(I_c\otimes U_j)$ also gives identical zero-counter output blocks to $8.87\times 10^{-14}$, verifying Proposition~\ref{prop:counterInvariance} in this compiled instance. %[L14]
\par\ifdefined\CPCReviewLayout\endgroup\fi

The amplification test checks the predicted dependence of success probability on the number of Grover iterations. For $N_t=4$ and $\tau^*=0.8$ (pre-amplification probability $p_0=\mathbb P_{\mathrm{TM},f}=0.140974$), applying $m=1,2,3$ full Grover iterations $\mathcal Q^m\mathcal W\ket{0^{11}}$ (up to $29{,}596$ elementary gates at $m=3$ in Table~\ref{tab:quantumCircuitAudit}) raises the success probability to $0.836625$ ($m=1$) and $0.880028$ ($m=2$), followed by the expected trigonometric over-rotation to $0.187085$ at $m=3$ (Fig.~\ref{fig:quantumVerification}(b)). These probabilities match $\sin^2((2m+1)\arcsin\sqrt{p_0})$ within $2\times 10^{-13}$ while keeping the normalized-state distance $\|\ket{\widetilde f^4}-\ket{f^4}\|_2\leq 1.02\times 10^{-14}$ across $m=0,1,2,3$ ($\leq 4.33\times 10^{-15}$ for $m=1,2$). %[L14]

The observable-readout test uses an $11$-qubit controlled Hadamard-test circuit with one readout ancilla on the $10$-qubit operator $\mathcal W_\phi$ in~\eqref{eq:linearObservableOverlap} and half-domain indicator $\ket w=\ket 0\otimes H^{\otimes 2}\ket{00}$. It recovers the terminal left-half scalar mass $\sum_{j=0}^3\phi_j^4=2.20206915759248$ against the classical value $2.20206915759257$ (relative error $4.15\times 10^{-14}$). This readout result is the exact ancilla expectation evaluated from the noiseless simulated statevector and excludes finite-shot sampling error; the archived gate count for $\mathcal W_\phi$ refers to the unconditioned $10$-qubit circuit rather than the $11$-qubit controlled Hadamard-test circuit. %[L14]

The grid-size tests in Fig.~\ref{fig:quantumVerification}(c) track compiled per-step \texttt{\{1Q, CNOT\}} counts and circuit depth across $N\in\{4,8,16,32,64\}$. Both grow approximately linearly over the tested range, consistent with the $\mathcal O(qN)$ cost of the directly compiled multiplexors. The uniform-steady branch uses constant velocity ($u=0.12$ for D1Q3), identical coefficient endpoints, and $\Lambda=1$, where the retained padding flag flips $F_{\mathrm{pad}}$ replace the physical ratio rotations. Separate D1Q3 ($N=8$) and D2Q5 ($4\times4$) audits verify the complete successful block against $B^0\oplus0$, including unused directions. Across the tested D1Q3 compilations, uniform-steady circuits require $4.41\%$--$13.22\%$ fewer CNOTs and $15.00\%$--$31.01\%$ less depth than general-coefficient circuits. This comparison uses complete circuits for both valid coefficient regimes and characterizes the stated direct-multiplexor compilation model. %[L14]

\subsection{History-state conditioning and padding diagnostics}\label{subsec:exp6History}

The final study separately examines history-matrix conditioning, terminal extraction with padding, and behavior near the positivity boundary. We assemble exact history matrices for 1D D1Q3 ($N_x=8$, $N_t\in\{4,8,16,32,64\}$, $T=0.5$, $\tau^*=0.8$, $c=\Delta x/\Delta t$, unsteady nonuniform flow) and 2D D2Q5 ($4\times 4$, $N_t\in\{4,8,16,32\}$, unsteady vortex flow) to test the bounds of Section~\ref{sec:QLSALBM} at the linear-algebra level. These assembled-matrix calculations do not simulate a compiled quantum linear systems algorithm (QLSA) circuit. %[L15]

\begin{figure*}[!htb]
 \centering
 \includegraphics[width=0.99\linewidth]{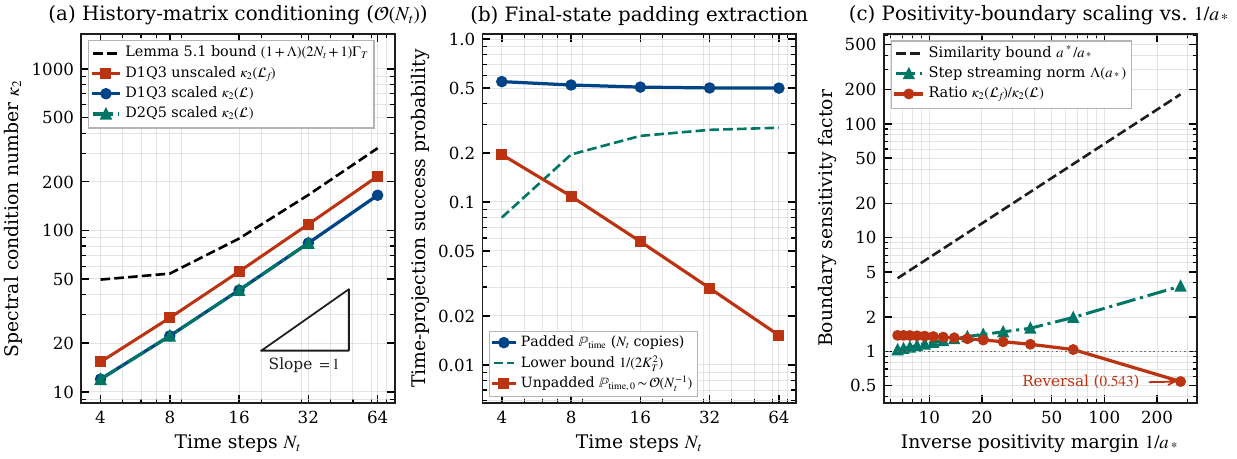}
 \caption{Exact history-matrix SVD conditioning and terminal-padding diagnostics. (a) Condition-number scaling $\mathcal O(N_t)$ for $\kappa_2(\mathcal L)$ and $\kappa_2(\mathcal L_f)$ on D1Q3 ($N_x=8$) and D2Q5 ($4\times 4$) compared against Lemma~\ref{lem:Lbound}. (b) Final-time projection probability $\mathbb P_{\mathrm{time}}$ on 1D D1Q3 under $N_t$-copy padding compared against unpadded $\mathcal O(N_t^{-1})$ decay and the theoretical lower bound $1/(2K_T^2)$ in~\eqref{eq:qlsaSuccess}. (c) Steady positivity-boundary scan ($N_x=8$, $N_t=12$, time-level weights identical) versus $1/a_*$ comparing the similarity factor $a^*/a_*$, coarse-grid streaming bound $\Lambda(a_*)$, and condition-number ratio $\kappa_2(\mathcal L_f)/\kappa_2(\mathcal L)$.}
 \label{fig:historyDiagnostics}
\end{figure*}

For conditioning, the spectral condition numbers of both padded history matrices grow linearly with $N_t$ over the tested time-refinement sequences: $\kappa_2(\mathcal L)=11.99\to 164.89$ on D1Q3 and $11.96\to 83.38$ on D2Q5 in Fig.~\ref{fig:historyDiagnostics}(a) and Table~\ref{tab:historyDiagnostics}. Lemma~\ref{lem:Lbound} bounds $\kappa_2(\mathcal L)$ within a factor below $2.1\times$ for all $N_t\geq 16$, and $\kappa_2(\mathcal L)<\kappa_2(\mathcal L_f)$ holds throughout these refinement sequences ($164.89$ versus $216.09$ at $N_t=64$). The finite sequence is consistent with the analytic $\mathcal O(N_t)$ bound supplied by Lemma~\ref{lem:Lbound}. %[L15]

For terminal extraction, direct solution of $\mathcal L\bb G=\bb F$ on the 1D D1Q3 sequence reproduces $\bb g^{N_t}$ to relative error at most $2.01\times 10^{-16}$. The $N_t$-copy padding keeps the final-time extraction probability at $\mathbb P_{\mathrm{time}}\in[0.5001,0.5483]>1/(2K_T^2)$, whereas the unpadded probability decreases from $0.1953$ to $0.0152$ with the observed $\mathcal O(N_t^{-1})$ behavior (Fig.~\ref{fig:historyDiagnostics}(b)). Thus padding addresses terminal-component extraction in these assembled systems independently of the condition-number comparison. %[L15]

The positivity-boundary scan tests a different issue on the fixed $N_x=8$ grid ($N_t=12$, $\tau^*=0.6$, $\Delta x=\Delta t=1/8$, $c=1$ in Fig.~\ref{fig:historyDiagnostics}(c)). Under steady spatially varying advection $u(x)=u_{\max}\sin(2\pi x)$, we scan $1/a_*\in[6.59,272.7]$ across $14$ values of $u_{\max}/c$; because $\mathcal A^{k+1}=\mathcal A^k$ while the spatial weights remain nonuniform, the assembled matrices directly test the similarity identity in Remark~\ref{rem:unscaledHistoryComparison}. The relative Frobenius residual $\|\mathcal L_f-\mathcal D\mathcal L\mathcal D^{-1}\|_F/\|\mathcal L_f\|_F$ does not exceed $1.12\times10^{-16}$, and both sides of~\eqref{eq:historySimilarityBounds} hold at every scan point. As $u_{\max}/c$ approaches $1/3$ without mesh refinement, the coarse-grid streaming bound $\Lambda(a_*)$ grows from $1.033$ to $3.744$, recording weaker normalization control near the positivity boundary. %[L15]

At the smallest tested positivity margin ($u_{\max}/c=0.326$), the condition-number ordering reverses to $\kappa_2(\mathcal L)=138.14$ versus $\kappa_2(\mathcal L_f)=75.03$ ($\kappa_2(\mathcal L_f)/\kappa_2(\mathcal L)=0.5432$). The reversal remains within the two-sided similarity envelope~\eqref{eq:historySimilarityBounds} and establishes that equilibrium energy scaling does not universally reduce the history-matrix condition number. %[L15]

\FloatBarrier
\section{Conclusions}\label{sec:conclusion}

Equilibrium energy scaling controls the normalization introduced when the full linear BGK distribution is evolved on a quantum computer. The inverse prescribed equilibrium weights define a time-dependent weighted norm in which local collision is nonexpansive for $\tau^*\geq 1/2$. In the scaled coordinates, any increase in the single-step norm comes from ratios of weights between the departure and arrival points of streaming. For uniform steady advection, the single-step normalization factor is one. At fixed physical time, a grid-independent cumulative bound follows for periodic permutation streaming when the equilibrium weights remain uniformly positive and the logarithmic directional derivatives along streaming characteristics admit a bound that is uniform in space, time, and mesh resolution. Because this time-dependent change of variables is exact at the discrete kinetic level, it preserves the multi-step non-equilibrium memory observed for $\tau^*\neq 1$, while the underlying LBM continuum discretization error and physical diffusive norm attenuation remain those of the original discrete dynamics. %[L01,L02,L03,L11]

This factorization yields a single-step block-encoding with two flag qubits, without uniform singular value amplification. A reusable failure counter then composes the steps coherently under a classical time schedule, after which the required terminal outputs are recovered. In the controlled $6\times6$, $\tau^*=0.8$ ablation, the same $50$-step evolution gains approximately $450\times$ in terminal success probability, reducing the corresponding ideal amplitude-amplification call factor by more than $21\times$. The comparison quantifies the reduction in sequential block-encoding normalization loss for the same discrete kinetic evolution. %[L04,L05,L12]

An end-to-end fault-tolerant implementation will require reversible evaluation of the coefficients, preparation of general non-equilibrium initial states, and certified fixed-point amplitude amplification. Resource accounting must include these operations together with the evolution circuit, extending the present tests of directly compiled multiplexors on small grids. A second direction is to extend the equilibrium-weighted construction to wall-boundary schemes whose streaming operators are not permutations and to nonlinear LBM models in which the velocity depends on the evolving distribution. %[L14,L17]

\appendix
\renewcommand{\thetheorem}{\Alph{section}.\arabic{theorem}}
\renewcommand{\thelemma}{\Alph{section}.\arabic{lemma}}
\renewcommand{\theproposition}{\Alph{section}.\arabic{proposition}}
\numberwithin{figure}{section}
\numberwithin{table}{section}
\renewcommand{\thefigure}{\Alph{section}.\arabic{figure}}
\renewcommand{\thetable}{\Alph{section}.\arabic{table}}

\section{Supporting model details and proofs}\label{app:proofs}

\subsection{Discrete velocities, block recurrence, and formal macroscopic consistency}\label{app:discreteDetails}\label{app:consistencyProof}
This subsection records the D2Q5 stencil (Fig.~\ref{fig:D2Q5}) and collision--streaming block recurrence used in Section~\ref{sec:LBM}, then supplies the fixed-$c$ macroscopic consistency argument needed to distinguish that formal result from the diffusive-refinement evidence. %[L01]
\begin{figure}[!htbp]
	\centering
	\includegraphics[width=0.62\columnwidth]{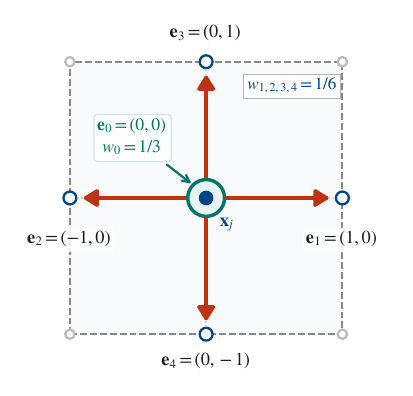}
	\caption{Discrete velocity stencil used by the two-dimensional D2Q5 lattice Boltzmann model. The index $i=0$ denotes the rest population, whereas $i=1,2,3,4$ denote the four axis-aligned nearest-neighbor streaming directions.}
	\label{fig:D2Q5}
\end{figure}

To expose the closed distribution update used in the main text, set $\alpha=1-1/\tau^*$, $\beta=1/\tau^*=1-\alpha$, $\mathcal A^k=\operatorname{diag}(A_0^k,\ldots,A_{q-1}^k)$, and $E_I=[I_N,\ldots,I_N]\in\mathbb R^{N\times qN}$. Each LBM step then decomposes into three operations, each of which is either local or a permutation:
\begin{enumerate}[(i)]
 \item \textbf{Local BGK collision:} $\bb f_i^{\mathrm{col},k}=\alpha\bb f_i^k+\beta A_i^k\bb\phi^k$, or globally %[L01]
 \begin{equation}\label{fcol}
  \bb f^{\mathrm{col},k}=\alpha\bb f^k+\beta\mathcal A^k E_I^\top\bb\phi^k,
 \end{equation}
 or, after substituting the scalar moment,
 \begin{equation}\label{matA}
\begin{aligned}
  \bb f^{\mathrm{col},k}&=\bigl(\alpha I_{qN}+\beta\mathcal A^k E_I^\top E_I\bigr)\bb f^k
  \\ &\text{upon substituting}\quad
  \bb\phi^k=\sum_{i=0}^{q-1}\bb f_i^k=E_I\bb f^k.
 \end{aligned}
\end{equation}
 \item \textbf{Exact lattice streaming:} Letting $P_i\ket j=\ket{\pi_i(j)}$ denote the periodic shift permutation mapping source site $\bb x_j$ to $\bb x_j+\bb c_i\Delta t$ ($P_0=I_N$) and $P=\operatorname{diag}(P_0,\ldots,P_{q-1})$, %[L01]
 \begin{equation}\label{fi}
  \bb f_i^{k+1}=P_i\bb f_i^{\mathrm{col},k},\qquad i=0,\ldots,q-1,
 \end{equation}
 so that
 \begin{equation}\label{fupdate}
  \bb f^{k+1}=P\bb f^{\mathrm{col},k}
  =P\bigl(\alpha\bb f^k+\beta\mathcal A^k E_I^\top\bb\phi^k\bigr).
 \end{equation}
 \item \textbf{Closed single-step distribution recurrence:} Combining \eqref{fupdate} with the zeroth-moment sum
 \begin{equation}\label{phiupdate}
  \bb\phi^k=E_I\bb f^k
 \end{equation}
 eliminates $\bb\phi^k$ as an intermediate state register and yields the closed time-dependent linear recurrence~\eqref{ftime}.
\end{enumerate}

The continuum equation corresponding to the discrete model is
\begin{equation}\label{eq_cde_LU}
 \frac{\partial\phi}{\partial t}+\nabla\cdot(\bb u\phi)
 =D\nabla^2\phi,
\end{equation}
whereas the linear equilibrium~\eqref{equibrium} omits quadratic velocity terms. Such quadratic terms are absent from linear passive-scalar equilibria unless anti-diffusive correction tensors or multiple-relaxation-time collisions are introduced \cite{ginzburg2005equilibrium,chai2013lattice,cartalade2019multiple}. Consequently, at finite lattice Mach number $\mathrm{Ma}=|\bb u|/c_s$, the discrete recurrence exhibits a well-characterized consistency deviation that depends on advection and is distinct from quantum circuit approximation errors \cite{kruger2017lattice,huang2014modified}. %[L16]

\begin{lemma}[Formal macroscopic consistency at fixed lattice speed]\label{lem:consistency}
 Suppose the distributions and prescribed velocity admit a smooth expansion as $\Delta t\to0$ with fixed $\tau^*$ and fixed lattice speed $c$, so that the Taylor remainders below are bounded. For the D2Q5 (or D1Q3) linear equilibrium \eqref{equibrium}, the modified equation through first order in $\Delta t$ is %[L16]
 \begin{equation}\label{eq:modifiedADE}
  \partial_t\phi+\nabla\cdot(\bb u\phi)
  =D\nabla^2\phi+
  \left(\tau-\frac{\Delta t}{2}\right)
        \partial_t\nabla\cdot(\bb u\phi)
  +\mathcal O(\Delta t^2).
 \end{equation}
\end{lemma}
The lemma isolates the classical discretization correction; quantum circuit error is analyzed separately.

The scope of Lemma~\ref{lem:consistency} is the fixed-$c$ modified equation. For constant $\bb u$, its leading advection-dependent correction is $-(\tau-\Delta t/2)(\bb u\cdot\nabla)^2\phi$ up to the stated $\mathcal O(\Delta t^2)$ remainder. In one dimension, this correction equals $-D\,\mathrm{Ma}^2\partial_x^2\phi$. This expression helps interpret the finite-lattice-unit comparisons in Section~\ref{subsec:exp1Gaussian}, though it does not provide an asymptotic error estimate at $\Delta x=\Delta t=1$. %[L16]

Under diffusive refinement, writing $h=\Delta x$ and $\Delta t=\gamma h^2$ with fixed $\gamma>0$ and $\tau^*>1/2$ gives $c=(\gamma h)^{-1}$ and $D=(\tau^*-1/2)/(3\gamma)$. If $\partial_t\nabla\cdot(\bb u\phi)$ remains uniformly bounded, the displayed correction $(\tau-\Delta t/2)\partial_t\nabla\cdot(\bb u\phi)$ is $\mathcal O(h^2)$. The fixed-$c$ remainder in Lemma~\ref{lem:consistency} is not assumed uniform as $c\to\infty$, so this observation alone does not establish second-order consistency or convergence of the complete scheme under diffusive scaling. The second-order convergence reported in Section~\ref{subsec:exp4Unsteady} is numerical evidence for the tested spatially uniform oscillatory velocities. All quantum stability and complexity results in this paper concern the exact discrete recurrence~\eqref{collision}. %[L16]

\begin{proof}[Proof of Lemma~\ref{lem:consistency}]
 Write $\mathscr D_i=\partial_t+\bb c_i\cdot\nabla$ and $f_i^{\mathrm{eq}}=a_i\phi$. Taylor expansion of \eqref{collision} yields %[L16]
 \[
  \begin{aligned}\mathscr D_i f_i+\frac{\Delta t}{2}\mathscr D_i^2 f_i&
  =-\frac{f_i-f_i^{\mathrm{eq}}}{\tau}+\mathcal O(\Delta t^2),
  \\
  f_i&=f_i^{\mathrm{eq}}-\tau\mathscr D_i f_i^{\mathrm{eq}}
       +\mathcal O(\Delta t^2).\end{aligned}
 \]
 Summing over $i$ and using zeroth-moment conservation $\sum_i(f_i-f_i^{\mathrm{eq}})=0$ gives
 \[
  \sum_i\mathscr D_i f_i^{\mathrm{eq}}
  =\left(\tau-\frac{\Delta t}{2}\right)
     \sum_i\mathscr D_i^2 f_i^{\mathrm{eq}}+\mathcal O(\Delta t^2).
 \]
 Substituting the exact equilibrium moments $\sum_i f_i^{\mathrm{eq}}=\phi$, $\sum_i\bb c_i f_i^{\mathrm{eq}}=\bb u\phi$, and $\sum_i\bb c_i\bb c_i^\top f_i^{\mathrm{eq}}=c_s^2\phi I$ evaluates the right-hand sum to $\partial_t^2\phi+2\partial_t\nabla\cdot(\bb u\phi)+c_s^2\nabla^2\phi$. Eliminating $\partial_t^2\phi$ via the differentiated leading-order balance $\partial_t\phi+\nabla\cdot(\bb u\phi)=\mathcal O(\Delta t)$ establishes \eqref{eq:modifiedADE}. The remainder estimates here hold at fixed lattice speed $c$. In particular, the final elimination uses $D=\mathcal O(\Delta t)$ and the leading-order advective balance. Under diffusive scaling, $D=\mathcal O(1)$, so neither that balance nor the fixed-$c$ remainder estimate can be transferred without a separate expansion. This proof establishes only the stated fixed-$c$ modified equation.
\end{proof} %[L16]

\subsection{Coefficient bounds and standing assumptions}\label{app:coefficientBounds}
For D2Q5, the prescribed velocity derivatives determine the coefficient-growth bound explicitly: $\mathscr D_0\log a_0=0$, $\mathscr D_{1,2}\log a_{1,2}=(3\partial_x u_x\pm(3/c)\partial_t u_x)/(1\pm 3u_x/c)$, and $\mathscr D_{3,4}\log a_{3,4}=(3\partial_y u_y\pm(3/c)\partial_t u_y)/(1\pm 3u_y/c)$. Under the velocity margin $|u_x|,|u_y|\leq(1/3-\delta_u)c$, these relations give the following explicit analytic bound without scanning the spatial grid: %[L03]
\[
 \ell_{\mathrm{st}}=\delta_u^{-1}\max\left\{
 \|\partial_x u_x\|_\infty+c^{-1}\|\partial_t u_x\|_\infty,\,
 \|\partial_y u_y\|_\infty+c^{-1}\|\partial_t u_y\|_\infty
 \right\}.
\]
Under diffusive refinement with the unit prefactor used in the refinement tests ($c=1/\Delta x\to\infty$), $c^{-1}\|\partial_t\bb u\|_\infty=\mathcal O(\Delta x)\to 0$ and $3|\bb u|/c=\mathcal O(\Delta x)\to 0$, so $\ell_{\mathrm{st}}$ remains uniformly bounded across successively finer meshes. For steady velocity fields ($\partial_t\bb u=\bb 0$), \eqref{eq:spaceTimeSmoothness} reduces to %[L03]
\begin{equation}\label{equ:continuousSmoothnessConstant}
 \max_{0\leq i<q}\sup_{\bb x\in\Omega}|\bb c_i\cdot\nabla\log a_i(\bb x)|\leq\ell,
\end{equation}
yielding
\begin{equation}\label{eq:steadyExponentialBound}
 \|B^k\|_2\leq\rho^k\leq\exp(\ell\Delta t/2),\qquad
 \|B^{k-1}\cdots B^0\|_2\leq\exp(\ell k\Delta t/2).
\end{equation}

\subsection{Failure-counter invariance and time-marching complexity proofs}\label{app:counterProof}\label{app:tmProof}

Measuring a step flag immediately prevents later coherent amplitude amplification, whereas retaining separate flags requires $aN_t$ qubits \cite{Low2019Interaction}. Proposition~\ref{prop:counterInvariance} replaces both choices with the reusable failure counter used in the main text. The counter-sector representation and the time-marching complexity proofs are recorded here. %[L07]

For zero-counter-controlled steps, the failed sectors take the staircase form below with $\nu_j=\prod_{r=1}^j\alpha_r$. The subsequent induction proves that uncontrolled steps have the same projected successful block, without asserting that their failed sectors retain this staircase form. %[L05]

\begin{equation}\label{flowmap}
\begin{gathered}
 \begin{bmatrix}\bb g^0\\\bb 0\\\hdashline\bb 0\\\bb 0\\\hdashline\bb 0\\\bb 0\\\hdashline\vdots\end{bmatrix}
 \xrightarrow{\mathcal V_1}
 \begin{bmatrix}\bb g^1/\nu_1\\\bb 0\\\hdashline\bb 0\\\bb h_1\\\hdashline\bb 0\\\bb 0\\\hdashline\vdots\end{bmatrix}
 \xrightarrow{\mathcal V_2}
 \begin{bmatrix}\bb g^2/\nu_2\\\bb 0\\\hdashline\bb 0\\\bb h_2\\\hdashline\bb 0\\\bb h_1\\\hdashline\vdots\end{bmatrix}
 \;\cdots\;\xrightarrow{\mathcal V_{N_t}}
 \begin{bmatrix}\bb g^{N_t}/\nu_{N_t}\\\bb 0\\\hdashline\bb 0\\\bb h_{N_t}\\\hdashline\bb 0\\\bb h_{N_t-1}\\\hdashline\vdots\\\hdashline\bb 0\\\bb h_1\end{bmatrix}.
\end{gathered}
\end{equation}

\begin{proof}[Proof of Proposition~\ref{prop:counterInvariance}]
 Let $\mathcal W_j\in\{\mathcal V_j,\widetilde{\mathcal V}_j\}$ and $\mathcal U_j^{\star}=\mathcal W_j\cdots\mathcal W_1$. We prove by induction on $j\in\{0,1,\ldots,N_t\}$ that for any system vector $\bb v$, %[L05]
 \[
  \mathcal U_j^{\star}(\ket{0^{n_c}}_c\ket{0^a}_a\otimes\bb v)
  =\ket{0^{n_c}}_c\ket{0^a}_a\otimes\left(\frac{B^{j-1}\cdots B^0}{\alpha_j\cdots\alpha_1}\bb v\right)
   +\sum_{c=1}^{j}\ket c_c\otimes\ket{\Psi_{j,c}},
 \]
 where the sum is empty at $j=0$. The base case $j=0$ is immediate. Assuming the representation holds at step $j-1$ ($1\leq j\leq N_t$), write $\bb w_{j-1}=(B^{j-2}\cdots B^0)\bb v/(\alpha_{j-1}\cdots\alpha_1)$. On the $c=0$ term, both $\overline U_j$ and $I_c\otimes U_j$ yield $\ket{0^{n_c}}_c\otimes\bigl(\ket{0^a}_a(B^{j-1}\bb w_{j-1}/\alpha_j)+\ket{\bot_j}\bigr)$ with $(\bra{0^a}_a\otimes I)\ket{\bot_j}=0$; applying $S$ fixes the $\ket{0^a}_a$ component at $c=0$ and shifts $\ket{0^{n_c}}_c\ket{\bot_j}$ to $\ket{1}_c\ket{\bot_j}$. On each failure term $\ket c_c\otimes\ket{\Psi_{j-1,c}}$ with $1\leq c\leq j-1$, both $\overline U_j$ and $I_c\otimes U_j$ act solely on the ancilla-system registers and preserve $\ket c_c$, producing $\ket c_c\otimes(\ket{0^a}_a\bb y_{j,c}+\ket{\bot'_{j,c}})$ with $(\bra{0^a}_a\otimes I)\ket{\bot'_{j,c}}=0$. Applying $S$ leaves the $\ket{0^a}_a$ part at counter $c\in\{1,\ldots,j-1\}$ and increments the $\ket{\bot'_{j,c}}$ part to counter $c+1\pmod{2^{n_c}}$. Because $1\leq c\leq j-1\leq N_t-1$ and $2^{n_c}\geq N_t+1$, we have $2\leq c+1\leq j\leq N_t<2^{n_c}$, so modular wrap-around to $c=0$ is impossible. This closes the induction; projecting onto $\bra{0^{n_c}}_c\bra{0^a}_a\otimes I$ at $j=N_t$ proves \eqref{eq:productBlock}.
\end{proof} %[L05]

\begin{proof}[Proof of Lemma~\ref{lem:recoverF}]
 Since $a_*\leq a_i(t_{N_t},\bb x_j)\leq 1$, rotating flag $a_{\mathrm{rec}}$ by $2\arccos(\sqrt{a_i(t_{N_t},\bb x_j)})$ gives zero-flag amplitude $\sqrt{a_i(t_{N_t},\bb x_j)}$, proving \eqref{eq:distributionRecovery}. The normalized error bound follows from $\|(\mathcal A^{N_t})^{1/2}(\ket{\widehat g^{N_t}}-\ket{g^{N_t}})\|\leq\varepsilon_g$, $\|(\mathcal A^{N_t})^{1/2}\ket{g^{N_t}}\|\geq\sqrt{a_*}$, and $\|u/\|u\|-v/\|v\|\|\leq 2\|u-v\|/\|v\|$.
\end{proof} %[L06]

The mass-based estimate used in Remark~\ref{rem:practicalBounds} follows directly from conservation and Cauchy--Schwarz:
 \[
  \begin{aligned}|M_{\mathrm{tot}}|
  &=\left|\sum_{j=0}^{N-1}\sum_{i=0}^{q-1}\sqrt{a_i(t_k,\bb x_j)}\,g_{i,j}^k\right|
  \\ &\leq\left(\sum_{j=0}^{N-1}\sum_{i=0}^{q-1}a_i(t_k,\bb x_j)\right)^{1/2}\|\bb g^k\|_2
  \\ &=\sqrt{N}\,\|\bb g^k\|_2.\end{aligned}
 \]

\begin{proof}[Proof of Theorem~\ref{thm:tmComplexity}]
 Set $\varepsilon_g=\varepsilon\sqrt{a_*}/8$. Letting $D_j=B^{j-1}/\Lambda$ and $\widetilde D_j$ denote the exact and implemented contractions ($\|\widetilde D_j-D_j\|_2\leq\eta$), telescoping gives $\|\widetilde D_{N_t}\cdots\widetilde D_1-D_{N_t}\cdots D_1\|_2\leq N_t\eta$. Writing $z=D_{N_t}\cdots D_1\ket{g^0}$ ($\|z\|=K_T^{-1}$) and $\widetilde z=\widetilde D_{N_t}\cdots\widetilde D_1\ket{g^0}$, choosing $\eta\leq\varepsilon_g/(8N_tK_T)$ ensures $\|\widetilde z-z\|\leq\varepsilon_g/(8K_T)$ and $\|\widetilde z/\|\widetilde z\|-z/\|z\|\|\leq 2\|\widetilde z-z\|/\|z\|\leq\varepsilon_g/4$. Allocating remaining synthesis errors within $\varepsilon_g$ and applying Lemma~\ref{lem:recoverF} with $\eta_{\mathrm{rec}}\leq\varepsilon\sqrt{a_*}/8$ keeps the normalized state error within $\varepsilon$ and success probability $\Omega(a_*/K_T^2)$. Amplitude amplification over $\mathcal O(K_T/\sqrt{a_*})=\mathcal O(K_T)$ rounds (for fixed $a_*$), each invoking $N_t$ step unitaries plus $U_0$ and $U_{\mathrm{rec}}$, establishes the stated query bounds.
\end{proof} %[L07]

\subsection{History-state solve and extraction cost}\label{app:qlsaProof}

To make the history-state comparison in the main text explicit, stack the unnormalized energy-scaled distribution vectors into $\bb G_0=[\bb g^0;\bb g^1;\cdots;\bb g^{N_t}]$. This converts \eqref{gtime} into the unpadded $(N_t+1)qN\times(N_t+1)qN$ lower-bidiagonal system %[L08]
\begin{equation}\label{eq:physicalHistorySystem}
\begin{aligned}
 &\mathcal L_0\bb G_0=\bb F_0,\\ &
 \mathcal L_0=
 \begin{bmatrix}
 I&&&&\\
 -B^{0}&I&&&\\
 &-B^{1}&I&&\\
 &&\ddots&\ddots&\\
 &&&-B^{N_t-1}&I
 \end{bmatrix},\qquad
 \bb F_0=\begin{bmatrix}\bb g^0\\\bb 0\\\vdots\\\bb 0\end{bmatrix},
\end{aligned}
\end{equation}
whose exact inverse collects all causal propagators:
\begin{equation}\label{eq:historyInverse}
 \mathcal L_0^{-1}=
 \begin{bmatrix}
 I&0&0&\cdots&0\\
 B^{0}&I&0&\cdots&0\\
 B^{1}B^{0}&B^{1}&I&\cdots&0\\
 \vdots&\vdots&\ddots&\ddots&\vdots\\
 B^{N_t-1}\cdots B^{0}&B^{N_t-1}\cdots B^{1}&\cdots&B^{N_t-1}&I
 \end{bmatrix}.
\end{equation}
Writing $J_0=I-\mathcal L_0$, we have $\mathcal L_0^{-1}=\sum_{r=0}^{N_t}J_0^r$ with $\|J_0^r\|_2\leq\Lambda^r\leq\Gamma_T$, giving $\sigma_{\max}(\mathcal L_0)\leq 1+\Lambda$ and $\sigma_{\min}(\mathcal L_0)=1/\|\mathcal L_0^{-1}\|_2\geq 1/((N_t+1)\Gamma_T)$. %[L08]

\begin{proof}[Proof of Lemma~\ref{lem:Lbound}]
 Let $J=I-\mathcal L$. Since $J$ shifts time labels by $+1$, $\|J\|_2=\max_{0\leq k<2N_t}\|B^k\|_2\leq\Lambda$ and $J^{2N_t+1}=0$, so $\mathcal L^{-1}=\sum_{r=0}^{2N_t}J^r$. Every nonzero block in $J^r$ is a product of at most $N_t$ physical matrices $B^k$ ($k<N_t$) padded by identities, with disjoint row/column time indices across blocks, implying $\|J^r\|_2\leq\Gamma_T$ and $\|\mathcal L^{-1}\|_2\leq(2N_t+1)\Gamma_T$.
\end{proof} %[L08]

The coefficient-access requirement of the history-state linear-system route appears explicitly in the following construction. Using coherent time-indexed coefficient access on an $n_{\mathrm{hist}}=\lceil\log_2(2N_t+1)\rceil$-qubit time register, we assemble the simultaneous step block-encoding %[L08,L09]
\begin{equation}\label{eq:hamT}
 (\bra{0^2}\otimes I)\mathsf{HAM\mbox{-}T}_{B}(\ket{0^2}\otimes I)
 =\sum_{k=0}^{2N_t-1}\ket k\bra k\otimes\frac{B^{k}}{\Lambda},
\end{equation}
which selects $U_B^k$ for $0\leq k<N_t$, a scalar ancilla rotation encoding $I/\Lambda$ for $N_t\leq k<2N_t$, and zero for $k\geq 2N_t$ at $\mathcal O(1)$ coefficient-oracle queries and $\mathcal O(\operatorname{poly}(\log N_t))$ time-comparison gates per call. Composing $\mathsf{HAM\mbox{-}T}_B$ with a unit-normalized block-encoding of the truncated shift $T_+=\sum_{k=0}^{2N_t-1}\ket{k+1}\bra k$ encodes $J/\Lambda=(T_+\otimes I)\sum_{k=0}^{2N_t-1}\ket k\bra k\otimes(B^k/\Lambda)$. A two-term LCU of $\frac{1}{1+\Lambda}I-\frac{\Lambda}{1+\Lambda}(J/\Lambda)$ then furnishes a $(1+\Lambda)$-normalized block-encoding of $\mathcal L$ at $\mathcal O(1)$ calls to $\mathsf{HAM\mbox{-}T}_B$ per query. Unused time indices $k>2N_t$ are extended by $I$ to preserve invertibility, and Hermitian dilation $H_{\mathcal L}=\bigl[\begin{smallmatrix}0&\mathcal L\\\mathcal L^\dagger&0\end{smallmatrix}\bigr]$ is applied whenever required by the linear-system routine. %[L08,L09]

To extract $\ket{g^{N_t}}$ from $\ket G=\bb G/\|\bb G\|$, prepare the uniform copy-window superposition $\ket{s_f}=\frac{1}{\sqrt{N_t+1}}\sum_{k=N_t}^{2N_t}\ket k$. Because $\|\bb G\|^2=\sum_{k=0}^{N_t}\|\bb g^k\|^2+N_t\|\bb g^{N_t}\|^2\leq(2N_t+1)\Gamma_T^2\|\bb g^0\|^2$, rank-one projection onto $\bra{s_f}$ yields $\sqrt{N_t+1}\,\bb g^{N_t}/\|\bb G\|$ with the time-extraction probability in~\eqref{eq:qlsaSuccess}. This removes the explicit $\mathcal O(N_t^{-1})$ history-length penalty and bounds the pre-error amplitude-amplification factor by %[L08,L09]
\begin{equation}\label{eq:qlsaAmplification}
 \mathcal O(\mathbb{P}_{\mathrm{time}}^{-1/2})
 =\mathcal O\left(\frac{\|\bb G\|}{\sqrt{N_t+1}\|\bb g^{N_t}\|}\right)
 =\mathcal O\left(\Gamma_T\frac{\|\bb g^0\|}{\|\bb g^{N_t}\|}\right)=\mathcal O(K_T).
\end{equation}
Composing time projection with physical recovery $U_{\mathrm{rec}}$ (Lemma~\ref{lem:recoverF}) yields combined success probability $\mathbb P_{\mathrm{QLSA},f}=\mathbb P_{\mathrm{time}}\mathbb P_{\mathrm{rec}}=(N_t+1)\|\bb f^{N_t}\|^2/\|\bb G\|^2\geq a_*/(2K_T^2)$. If $\ket{\widetilde G}$ approximates $\ket G$ within $\|\ket{\widetilde G}-\ket G\|\leq\varepsilon_h$, the unnormalized copy-averaged vector $\widetilde{\bb g}_{\mathrm{avg}}^{\,N_t}=(N_t+1)^{-1/2}(\bra{s_f}\otimes I)(\|\bb G\|\ket{\widetilde G})$ satisfies $\|\widetilde{\bb g}_{\mathrm{avg}}^{\,N_t}-\bb g^{N_t}\|\leq\sqrt{2}\varepsilon_h\Gamma_T\|\bb g^0\|$, and the normalized recovered physical state error is bounded by $4\sqrt{2}K_T\varepsilon_h/\sqrt{a_*}$ prior to recovery gate errors. %[L08,L09]

\begin{proof}[Proof of Theorem~\ref{thm:qlsaComplexity}]
 By Lemma~\ref{lem:Lbound} and \cite{Costa2021QLSA}, preparing $\ket{\widetilde G}$ to error $\varepsilon_h=\mathcal O(\varepsilon\sqrt{a_*}/K_T)$ with constant probability costs $\mathcal O\bigl((1+\Lambda)\|\mathcal L^{-1}\|_2\log(1/\varepsilon_h)\bigr)=\mathcal O\bigl((1+\Lambda)N_t\Gamma_T\log(K_T/\varepsilon)\bigr)$ queries to the $\mathcal L$-block-encoding and $\ket 0\ket{g^0}$ preparation (for fixed $a_*$). Time projection onto $\ket{s_f}$ and Lemma~\ref{lem:recoverF} produce $\ket{\widetilde f^{N_t}}$ within error $\varepsilon$ at combined success probability $\Omega(a_*/K_T^2)$; $\mathcal O(K_T)$ rounds of amplitude amplification yield \eqref{eq:qlsaFinalCost}. Perturbing $\mathcal L$ by $E$ with $\|\mathcal L^{-1}E\|_2<1/2$ induces relative history error $\leq 2\|\mathcal L^{-1}\|_2\|E\|_2$, verifying the stated block-precision requirement.
\end{proof} %[L08,L09]

For an explicit precision allocation, one can take $\varepsilon_h\leq\varepsilon\sqrt{a_*}/(32\sqrt{2}K_T)$ and $b=\mathcal O\bigl(\log\frac{(1+\Lambda)N_t\Gamma_T K_T}{\varepsilon\sqrt{a_*}}\bigr)$ arithmetic bits for coefficient values, square roots and streaming ratios.
\section{Explicit gate-level synthesis of structured QLBM modules}\label{app:circuitSynthesis} %[L09]

This appendix records the exact elementary-gate synthesis using directly compiled multiplexed rotations over \texttt{\{X, H, Z, RY, RZ, P, CNOT\}} implemented in Section~\ref{subsec:exp5Quantum}. This explicit implementation uses no additional scratch ancillas; the reversible-arithmetic oracle model remains a distinct cost model. Here $R_Y(\theta)=\exp(-\i\theta Y/2)$, $R_Z(\theta)=\exp(-\i\theta Z/2)$, and $P(\varphi)=\operatorname{diag}(1,\e^{\i\varphi})$. %[L17]

\subsection{Rotation primitives and local equilibrium preparation}

In this implementation, any directly compiled $m$-control multiplexed rotation $\sum_{z=0}^{2^m-1}\ket z\bra z\otimes R_\sigma(\theta_z)$ ($\sigma\in\{Y,Z\}$) decomposes via a Gray-code Walsh--Hadamard angle transform into at most $2^m$ single-qubit $R_\sigma$ rotations interleaved with $2^m$ \texttt{CNOT} gates and requires zero scratch ancillas \cite{barenco1995elementary,shende2006synthesis}. Recursively pairing an $m$-control $R_Z$ multiplexor with an $(m-1)$-control phase gate yields an exact $m$-controlled phase gate \texttt{MCP}$(\varphi)$ (giving \texttt{MCZ} at $\varphi=\pi$ and \texttt{MCX} via $H\text{--}\mathsf{MCZ}\text{--}H$ on the target) using $2^{m+1}-2$ \texttt{CNOT}s, $2^{m+1}-1$ single-qubit rotations, and at most $2m$ \texttt{X} gates for open zero-controls. Big-endian modular increment $\mathsf{ADD}$ and inverse decrement $\mathsf{ADD}^\dagger$ on an $r$-qubit register cascade \texttt{MCX} gates from the most significant bit to the least significant bit and in reverse order, respectively. %[L17]

To ensure that $R_\Pi^k=V^k((2\ket{0^{n_d}}\bra{0^{n_d}}-I_{2^{n_d}})\otimes I_N)(V^k)^\dagger$ acts strictly on the physical subspace without leaking amplitude into unused direction labels ($i\geq q$), we synthesize $V^k$ as a short cascade of $R_Y$ multiplexors that leaves $\ket i_d$ ($i\geq q$) invariant at every node $j\in[N]$ (writing $a_i:=a_i(t_k,\bb x_j)$):
\begin{itemize}
 \item \emph{D1Q3 ($q=3$, $n_d=2$, qubits $(d_1,d_0)$, unused state $\ket{11}_d$):} Layer~1 rotates $d_0$ by $\theta_{1,j}=2\arcsin(\sqrt{a_1})$ controlled on $d_1=0$ and multiplexed over $g$. Layer~2 rotates $d_1$ by $\theta_{2,j}=2\arcsin(\sqrt{a_2/(a_0+a_2)})$ controlled on $d_0=0$ and multiplexed over $g$. On $\ket{00}_d\ket j_g$, Layer~1 produces $(\sqrt{a_0+a_2}\ket{00}_d+\sqrt{a_1}\ket{01}_d)\ket j_g$ and Layer~2 splits $\ket{00}_d$ into $\sqrt{a_0}\ket{00}_d+\sqrt{a_2}\ket{10}_d$, whereas on $\ket{11}_d\ket j_g$ both controls fail and fix $\ket{11}_d\ket j_g$ identically using two $(1+n_g)$-control $R_Y$ multiplexors ($4N$ $R_Y$ and $4N$ \texttt{CNOT} gates).
 \item \emph{D2Q5 ($q=5$, $n_d=3$, qubits $(d_2,d_1,d_0)$, unused states $\ket{101}_d,\ket{110}_d,\ket{111}_d$):} Layer~1 rotates $d_2$ by $2\arcsin(\sqrt{a_4})$ controlled on $(d_1,d_0)=(0,0)$ ($4N$ $R_Y$ and $4N$ \texttt{CNOT}). Layer~2 rotates $d_1$ by $2\arcsin(\sqrt{(a_2+a_3)/(1-a_4)})$ controlled on $d_2=0$ ($2N$ $R_Y$ and $2N$ \texttt{CNOT}). Layer~3 rotates $d_0$ controlled on $d_2=0$ by $2\arcsin(\sqrt{a_1/(a_0+a_1)})$ when $d_1=0$ and $2\arcsin(\sqrt{a_3/(a_2+a_3)})$ when $d_1=1$ ($4N$ $R_Y$ and $4N$ \texttt{CNOT}). Because every unused basis state has $d_2=1$ and $(d_1,d_0)\neq(0,0)$, all three layers act trivially on $\{\ket{101}_d,\ket{110}_d,\ket{111}_d\}$ ($10N$ $R_Y$ and $10N$ \texttt{CNOT} gates total).
\end{itemize} %[L17]

\subsection{Single-step realization}
The collision flag uses angle $\theta_\alpha=2\arcsin\sqrt{1/(2\tau^*)}$. The exact unitary is
{\footnotesize
\[
 \begin{aligned}U_C^k&=\bigl(R_Y(-\theta_\alpha)_{a_C}\otimes I\bigr)\bigl(I_2\otimes V^k\bigr)
 \\ &\quad\Bigl(\ket 0\bra 0_{a_C}\otimes I+\ket 1\bra 1_{a_C}\otimes(2\ket{0^{n_d}}\bra{0^{n_d}}-I_{2^{n_d}})\otimes I_N\Bigr)
 \\ &\quad\bigl(I_2\otimes(V^k)^\dagger\bigr)\bigl(R_Y(\theta_\alpha)_{a_C}\otimes I\bigr),\end{aligned}
\]
}
where the central phase uses $Z_{a_C}$ and an open-zero \texttt{MCZ} on $(a_C=1,d=0^{n_d})$. Because the neighboring $V^k$ and $(V^k)^\dagger$ cancel on the $\ket 0_{a_C}$ flag branch, neither requires ancilla control. Streaming applies the ratio-rotation unitary $D_\Lambda^k$ followed by $P_{\mathrm{ext}}$, with $\pi$ rotations on unused directions. For uniform steady coefficients with $\Lambda=1$, $D_\Lambda^k$ can be replaced by $F_{\mathrm{pad}}$ with the same successful block. Complete module counts appear in Table~\ref{tab:moduleResources}. %[L17]

\begin{proof}[Proof of Proposition~\ref{prop:structuredBE}]
 Because $U_C^k$ acts on $(a_C,d,g)$ and $U_P^k$ acts on $(a_P,d,g)$, projecting onto $\bra{0}_{a_P}\bra{0}_{a_C}$ factors directly into $(\widetilde P^k/\Lambda)\widetilde C^k=B^k/\Lambda$ on the physical subspace and $0\cdot(\alpha I)=0$ on unused direction states $i\geq q$. Standard propagation of $\mathcal O(\eta)$ rotation-angle errors yields the finite-precision bound.
\end{proof} %[L17,L04]

\subsection{Module resource summary and linear-functional circuit}\label{app:completeAlgorithm}\label{app:readoutCircuit}

\begin{table*}[!htbp]
\centering\small
\caption{Quantum register layout and single-step module costs for coherent energy-scaled QLBM time marching in distribution mode, shown separately for the analytic reversible-oracle model and the implementation that directly compiles multiplexed rotations ($n_d=\lceil\log_2 q\rceil$, $n_g=\lceil\log_2 N\rceil$, $n_c=\lceil\log_2(N_t+1)\rceil$, $b=\mathcal O(\log(N_tK_T/(\varepsilon\sqrt{a_*})))$).}
\label{tab:moduleResources}
\setlength{\tabcolsep}{5pt}
\begin{tabular}{>{\raggedright\arraybackslash}p{3.4cm}>{\raggedright\arraybackslash}p{2.3cm}>{\raggedright\arraybackslash}p{4.3cm}>{\raggedright\arraybackslash}p{4.8cm}}
\toprule
Module / Register & Qubit width & Analytic $b$-bit oracle model & Explicit multiplexor synthesis (\ref{app:circuitSynthesis}) \\
\midrule
Registers $(c,a_{\mathrm{rec}},a_P,a_C,d,g)$ & $n_c+3+n_d+n_g$ ($+ \mathcal O(b)$ work wires in oracle model) & Logarithmic in $N$ and $N_t$: $\lceil\log_2 N\rceil+\lceil\log_2(N_t+1)\rceil+\mathcal O(b)$ & For the implementation using directly compiled multiplexed rotations, zero arithmetic scratch qubits: exactly $n_c+3+n_d+n_g$ logical qubits \\
Collision LCU $U_C^k$ & Acts on $(a_C,d,g)$ ($1+n_d+n_g$) & $2$ calls to $V^k/(V^k)^\dagger$ ($\mathcal O(q)$ calls to $O_a$), $\mathcal O(\operatorname{poly}(b))$ gates & $4$ (D1Q3) or $6$ (D2Q5) $R_Y$ multiplexors $+1$ \texttt{MCZ} ($\mathcal O(qN)$ \texttt{\{1Q, CNOT\}} gates) \\
Weighted streaming $U_P^k=(I_2\otimes P_{\mathrm{ext}})D_\Lambda^k$ & Acts on $(a_P,d,g)$ ($1+n_d+n_g$) & $2$ calls to $O_a$ at $(t_k,t_{k+1})$ ($0$ if steady uniform) $+2d$ modular shifts, $\mathcal O(\operatorname{poly}(n_g,b))$ gates & $1$ $(n_d+n_g)$-controlled $R_Y$ multiplexor (replaced by direction-only padding flag flips if steady uniform) $+2d$ modular shifts ($\mathcal O(qN+dn_g 2^{n_g})$ gates) \\
Failure counter $S$ & Acts on $(c,a_P,a_C)$ ($n_c+2$) & $1$ increment $+1$ 2-controlled decrement, $\mathcal O(n_c^2)$ gates & Exact zero-ancilla \texttt{MCX} cascade ($\mathcal O(n_c 2^{n_c})=\mathcal O(N_t\log N_t)$ gates, or $\mathcal O(n_c)$ with $1$ clean ancilla) \\
Terminal recovery $U_{\mathrm{rec}}$ or $(V^{N_t})^\dagger$ & Applied once at $k=N_t$ & $\mathcal O(1)$ calls to $O_a$ at $t_{N_t}$, $\mathcal O(\operatorname{poly}(b))$ gates & $1$ $R_Y$ multiplexor for $U_{\mathrm{rec}}$, or $2$--$3$ multiplexors for $(V^{N_t})^\dagger$ \\
\bottomrule
\end{tabular}
\end{table*}

Table~\ref{tab:moduleResources} separates the register widths and per-step module costs of the analytic reversible-oracle model from those of the implementation that directly compiles multiplexed rotations for Algorithm~\ref{alg:timeMarching} in Section~\ref{subsec:growthAccuracy}. The table is specialized to distribution output and a constant-factor estimate $\widehat K_T=\mathcal O(K_T)$; looser certified bounds remain explicit in Algorithm~\ref{alg:timeMarching}. %[L17]

For the linear-functional output task in Section~\ref{subsec:readout}, the unamplified circuit that evaluates the overlap with the weight state is %[L17]
 \[
  \mathcal W_\phi=\bigl(I_{c,a,d}\otimes U_w^\dagger\bigr)\bigl(I_{c,a}\otimes(V^{N_t})^\dagger\bigr)\widetilde{\mathcal U}_{N_t}\bigl(I_{c,a}\otimes V^0\bigr)\bigl(I_{c,a,d}\otimes U_{\phi^0}\bigr).
 \]
Its all-zero matrix element evaluates~\eqref{eq:linearObservableOverlap} without preparing a normalized scalar state.

\section{Additional verification data for steady transport, operator norms, and history matrices}\label{app:errorBudget}

The following tables provide the complete steady-transport results, operator-norm audit, and history-matrix diagnostics used by the corresponding numerical arguments in Section~\ref{sec:numerics}. %[L10,L13,L15]

\begin{table*}[!htbp]
\centering\small
\caption{Accuracy metrics, scaled-solution norm ratios $R_T=\|\bb g^0\|/\|\bb g^k\|$, and exact pre-amplification quantum success probabilities ($\Lambda=\Gamma_T=1$) for the 1D D1Q3 ($N_x=128$) and 2D D2Q5 ($64\times 64$) steady periodic Gaussian benchmarks.}
\label{tab:gaussianCombined}
\begin{tabular}{llrrrrrrrr}
\toprule
Model & $\tau^*$ & $k$ & $E_{\mathrm{ADE}}^k$ & $E_{\mathrm{eq}}^k$ & $E_{\mathrm{mass}}^k$ & $R_T$ & $\mathbb P_{\mathrm{TM}}$ & $\mathbb P_{\mathrm{TM},f}$ & $\mathbb P_\phi$ \\
\midrule
1D D1Q3 & 0.8 & 20 & 7.28e-04 & 8.31e-16 & 2.20e-15 & 1.0038 & 0.9925 & 0.5161 & 0.9996 \\
        & 0.8 & 40 & 1.61e-03 & 1.32e-15 & 4.41e-15 & 1.0076 & 0.9849 & 0.5121 & 0.9996 \\
        & 1.0 & 20 & 1.50e-03 & 3.66e-16 & 1.26e-15 & 1.0061 & 0.9878 & 0.5137 & 0.9994 \\
        & 1.0 & 40 & 2.92e-03 & 7.51e-16 & 2.20e-15 & 1.0125 & 0.9755 & 0.5073 & 0.9994 \\
        & 1.3 & 20 & 2.80e-03 & 3.37e-16 & 6.30e-16 & 1.0095 & 0.9813 & 0.5103 & 0.9989 \\
        & 1.3 & 40 & 4.97e-03 & 4.80e-16 & 1.26e-15 & 1.0195 & 0.9622 & 0.5004 & 0.9990 \\
\midrule
2D D2Q5 & 0.8 & 10 & 5.99e-03 & 1.96e-15 & 1.21e-15 & 1.0328 & 0.9374 & 0.2457 & 0.9930 \\
        & 0.8 & 30 & 1.87e-02 & 5.78e-15 & 3.32e-15 & 1.0989 & 0.8281 & 0.2171 & 0.9938 \\
        & 1.0 & 10 & 1.25e-02 & 1.26e-15 & 6.03e-16 & 1.0515 & 0.9045 & 0.2372 & 0.9895 \\
        & 1.0 & 30 & 3.11e-02 & 3.87e-15 & 1.81e-15 & 1.1575 & 0.7464 & 0.1957 & 0.9913 \\
        & 1.3 & 10 & 2.40e-02 & 6.52e-16 & 7.54e-16 & 1.0763 & 0.8633 & 0.2266 & 0.9832 \\
        & 1.3 & 30 & 4.81e-02 & 1.81e-15 & 1.81e-15 & 1.2377 & 0.6528 & 0.1713 & 0.9873 \\
\bottomrule
\end{tabular}
\end{table*}

\begin{table*}[!htbp]
\centering\small
\caption{Explicit 2D D2Q5 operator-norm audit ($\tau^*=1.3$, $\Delta t=1/n^2$) for uniform, steady nonuniform ($t_k=0$), and unsteady nonuniform ($t_k=0.125$) flows, reporting the unscaled norm $\|M^k\|_2$, theoretical bound $\rho^k$ (Theorem~\ref{thm:norm}), energy-scaled norm $\|B^k\|_2$, and relative Frobenius factorization residual $E_{\mathrm{fac}}=\|B^k-\widetilde P^k\widetilde C^k\|_F/\|B^k\|_F$.}
\label{tab:matrixValidation}
\begin{tabular}{rlrrrr}
\toprule
$n$ & Velocity regime & $\|M^k\|_2$ & $\rho^k$ & $\|B^k\|_2$ & $E_{\mathrm{fac}}$ \\
\midrule
4 & uniform ($t_k=0$)               & 1.034786 & 1.000000 & 1.000000 & $<10^{-16}$ \\
4 & steady nonuniform ($t_k=0$)     & 1.035236 & 1.016622 & 1.000319 & $<10^{-16}$ \\
4 & unsteady nonuniform ($t_k=0.125$) & 1.034627 & 1.021897 & 1.000432 & $<10^{-16}$ \\
8 & uniform ($t_k=0$)               & 1.034210 & 1.000000 & 1.000000 & $<10^{-16}$ \\
8 & steady nonuniform ($t_k=0$)     & 1.034322 & 1.005556 & 1.000040 & $<10^{-16}$ \\
8 & unsteady nonuniform ($t_k=0.125$) & 1.034170 & 1.004982 & 1.000018 & $<10^{-16}$ \\
\bottomrule
\end{tabular}
\end{table*}

\begin{table*}[!htbp]
\centering\small
\caption{Exact history-matrix conditioning, Lemma~\ref{lem:Lbound} upper bounds, and final-time extraction probabilities across $N_t$ ($T=0.5$, $\tau^*=0.8$, $c=\Delta x/\Delta t$) for 1D D1Q3 ($N_x=8$) and 2D D2Q5 ($4\times 4$) unsteady nonuniform flows.}
\label{tab:historyDiagnostics}
\begin{tabular}{lrrrrrrrr}
\toprule
Model & $N_t$ & $\Gamma_T$ & $\kappa_2(\mathcal L_0)$ & $\kappa_2(\mathcal L)$ & $\kappa_2(\mathcal L_f)$ & Lemma~\ref{lem:Lbound} & $\mathbb P_{\mathrm{time},0}$ & $\mathbb P_{\mathrm{time}}$ ($N_t$ copies) \\
\midrule
1D D1Q3 &  4 & 2.446 &  6.79 &  11.99 &  15.32 &  49.55 & 0.1953 & 0.5483 \\
        &  8 & 1.545 & 11.95 &  22.21 &  28.73 &  53.98 & 0.1084 & 0.5225 \\
        & 16 & 1.335 & 22.20 &  42.62 &  55.55 &  88.93 & 0.0571 & 0.5072 \\
        & 32 & 1.269 & 42.61 &  83.39 & 109.09 & 165.62 & 0.0296 & 0.5012 \\
        & 64 & 1.248 & 83.38 & 164.89 & 216.09 & 322.53 & 0.0152 & 0.5001 \\
\midrule
2D D2Q5 &  4 & 1.379 & ---   &  11.96 &  12.67 &  25.86 & ---    & ---    \\
        &  8 & 1.212 & ---   &  22.20 &  23.60 &  41.70 & ---    & ---    \\
        & 16 & 1.155 & ---   &  42.61 &  45.46 &  76.59 & ---    & ---    \\
        & 32 & 1.136 & ---   &  83.38 &  89.12 & 147.99 & ---    & ---    \\
\bottomrule
\end{tabular}
\end{table*}

\section{Periodic Gaussian reference solution and steady-transport test settings}\label{app:benchmarkDetails}

The accompanying Python suite reproduces the six tests. This appendix specifies the continuum reference and initialization used by the steady and uniform oscillatory transport comparisons. %[L10,L13]

On a $d$-dimensional periodic box of side lengths $\bb L=(L_1,\ldots,L_d)$ with spatially uniform velocity $\bb u(\bb x,t)=\bb u(t)$ and accumulated displacement $\bb s(t)=\int_0^t\bb u(\xi)\,\d\xi$ ($\bb s(t)=\bb u t$ for steady flow), the exact continuum ADE solution for an initial Gaussian hill of amplitude $\phi_0$ and width $\sigma_0$ centered at $\bb x_0$ is the periodized sum %[L10,L13]
\begin{equation}\label{eq:periodicGaussian}
\begin{aligned}
 \phi_{\mathrm{ADE}}(\bb x,t)
 &=\phi_0\left(\frac{\sigma_0^2}{\sigma_0^2+2Dt}\right)^{d/2}
 \sum_{\bb z\in\mathbb Z^d}
 \\ &\quad\exp\left(-\frac{\|\bb x-\bb x_0-\bb s(t)+\bb z\odot\bb L\|_2^2}{2(\sigma_0^2+2Dt)}\right),
\end{aligned}
\end{equation}
which in numerical evaluations is truncated to $\bb z\in\{-3,\ldots,3\}^d$ with remainder below IEEE-754 double-precision roundoff, while the lattice distribution is initialized at local equilibrium $f_i^0(\bb x_j)=a_i(0,\bb x_j)\phi_{\mathrm{ADE}}(\bb x_j,0)$ with %[L10,L13]
\begin{equation}\label{equ:initial_condition}
 \phi^0(\bb x)=\phi_{\mathrm{ADE}}(\bb x,0).
\end{equation}

The steady uniform advection benchmarks in Section~\ref{subsec:exp1Gaussian} use lattice units ($\Delta x=\Delta t=1$, giving $\Lambda=\Gamma_T=1$ identically) across over-relaxation ($\tau^*=0.8$, $D=0.1$), full relaxation ($\tau^*=1.0$, $D=1/6$), and under-relaxation ($\tau^*=1.3$, $D=4/15$). Specifically, the 1D D1Q3 runs use $N_x=128$ nodes ($x_0=64$, $\sigma_0=15$, $\phi_0=0.3$, $u=0.2$, $a_*=1/15$) evaluated at $k=20,40$, and the 2D D2Q5 runs use $64\times 64$ nodes ($\bb x_0=(32,32)$, $\sigma_0=5$, $\phi_0=0.3$, $\bb u=(0.2,0.2)$, $a_*=1/15$) evaluated at $k=10,30$. %[L10]

\FloatBarrier
\section*{CRediT authorship contribution statement}
\textbf{Yuan He:} Conceptualization, Methodology, Software, Formal analysis, Validation, Visualization, Writing -- original draft.
\textbf{Yuan Yu:} Conceptualization, Methodology, Formal analysis, Supervision, Funding acquisition, Writing -- review \& editing.
\textbf{Yue Yu:} Methodology, Formal analysis, Validation, Funding acquisition, Writing -- review \& editing.
\textbf{Yuting Zhou:} Software, Validation, Visualization, Writing -- review \& editing.

\section*{Declaration of competing interest}
The authors declare that they have no known competing financial interests or personal relationships that could have appeared to influence the work reported in this paper.

\section*{Data availability}
Data and source code supporting the findings of this study are available from the corresponding author upon reasonable request.

\section*{Acknowledgements}
Yuan Yu was supported by the National Natural Science Foundation of China (No.\ 12101527) and the Natural Science Foundation of Hunan Province (No.\ 2026JJ60005).
Yue Yu was supported by NSFC grant (No.\ 12301561), the Key Project of Scientific Research Project of Hunan Provincial Department of Education (No.\ 24A0100), the Science and Technology Innovation Program of Hunan Province (No.\ 2025RC3150), the general program of Hunan Provincial Natural Science Foundation (No.\ 2026JJ50003), and partially supported by NSFC grant No.\ 12341104.
This research was supported in part by the Science and Technology Innovation Program of Hunan Province (No.\ 2026RC3172), the 111 Project (No.\ D23017), and Program for Science and Technology Innovative Research Team in Higher Educational Institutions of Hunan Province of China. Computational resources were provided by the High Performance Computing Platform of Xiangtan University.

\section*{Declaration of generative AI and AI-assisted technologies in the manuscript preparation process}
OpenAI ChatGPT/Codex was used during manuscript preparation to assist with organizing the manuscript, drafting and revising paragraphs, and improving clarity and language. The authors reviewed and edited all AI-assisted material and take full responsibility for the scientific accuracy, references, and conclusions of this article.

\FloatBarrier
\bibliographystyle{elsarticle-num}
\bibliography{RefsQLBM}

@article{chan1984stability,
	title={Stability analysis of finite difference schemes for the advection-diffusion equation},
	author={T. F. Chan},
	journal={SIAM J. Numer. Anal.},
	volume={21},
	number={2},
	pages={272--284},
	year={1984}
}

@article{dehghan2004numerical,
	title={Numerical solution of the three-dimensional advection-diffusion equation},
	author={M. Dehghan},
	journal={Appl. Math. Comput.},
	volume={150},
	number={1},
	pages={5--19},
	year={2004}
}

@article{ginzburg2005equilibrium,
	title={Equilibrium-type and link-type lattice {B}oltzmann models for generic advection and anisotropic-dispersion equation},
	author={I. Ginzburg},
	journal={Adv. Water Resour.},
	volume={28},
	number={11},
	pages={1171--1195},
	year={2005}
}

@book{kruger2017lattice,
	title={The Lattice {B}oltzmann Method},
	author={T. Kr\"{u}ger and H. Kusumaatmaja and A. Kuzmin and O. Shardt and G. Silva and E. M. Viggen},
	year={2017},
	publisher={Springer}
}

@inproceedings{Gilyen2019QSVD,
	title={Quantum singular value transformation and beyond: exponential improvements for quantum matrix arithmetics},
	author={A. Gily\'{e}n and Y. Su and G. H. Low and N. Wiebe},
	booktitle={Proceedings of the 51st Annual ACM SIGACT Symposium on Theory of Computing},
	pages={193--204},
	year={2019}
}

@inproceedings{Chakraborty2019blockEncode,
	title={The power of block-encoded matrix powers: improved regression techniques via faster {H}amiltonian simulation},
	author={S. Chakraborty and A. Gily\'{e}n and S. Jeffery},
	booktitle={46th International Colloquium on Automata, Languages, and Programming (ICALP 2019)},
	series={Leibniz International Proceedings in Informatics (LIPIcs)},
	volume={132},
	pages={33:1--33:16},
	year={2019},
	doi={10.4230/LIPIcs.ICALP.2019.33}
}

@article{Low2019Interaction,
	title={Hamiltonian simulation in the interaction picture},
	author={G. H. Low and N. Wiebe},
	journal={Quantum},
	volume={3},
	pages={163},
	year={2019},
	doi={10.22331/q-2019-07-12-163}
}

@article{todorova2020quantum,
	title={Quantum algorithm for the collisionless {B}oltzmann equation},
	author={B. N. Todorova and R. Steijl},
	journal={J. Comput. Phys.},
	volume={409},
	pages={109347},
	year={2020}
}

@article{budinski2021quantum,
	title={Quantum algorithm for the advection-diffusion equation simulated with the lattice {B}oltzmann method},
	author={L. Budinski},
	journal={Quantum Inf. Process.},
	volume={20},
	number={2},
	pages={57},
	year={2021}
}

@article{Costa2021QLSA,
	title={Optimal scaling quantum linear-systems solver via discrete adiabatic theorem},
	author={P. C. S. Costa and D. An and Y. R. Sanders and Y. Su and R. Babbush and D. W. Berry},
	journal={PRX Quantum},
	volume={3},
	number={4},
	pages={040303},
	year={2022}
}

@article{Lin2023timeODE,
	title={Time-marching based quantum solvers for time-dependent linear differential equations},
	author={D. Fang and L. Lin and Y. Tong},
	journal={Quantum},
	volume={7},
	pages={955},
	year={2023}
}

@article{ACL2023LCH2,
	title={Quantum algorithm for linear non-unitary dynamics with near-optimal dependence on all parameters},
	author={D. An and A. M. Childs and L. Lin},
	journal={arXiv:2312.03916},
	year={2023}
}

@article{kocherla2024fully,
	title={Fully quantum algorithm for mesoscale fluid simulations with application to partial differential equations},
	author={S. Kocherla and Z. Song and F. E. Chrit and B. Gard and E. F. Dumitrescu and A. Alexeev and S. H. Bryngelson},
	journal={AVS Quantum Sci.},
	volume={6},
	number={3},
	year={2024}
}

@article{wawrzyniak2024unitary,
	title={Unitary quantum algorithm for the lattice-{B}oltzmann method},
	author={D. Wawrzyniak and J. Winter and S. Schmidt and T. Indinger and U. Schramm and C. Jan{\ss}en and N. A. Adams},
	journal={arXiv:2405.13391},
	year={2024}
}

@article{wawrzyniak2025linearized,
	title={Linearized quantum lattice-{B}oltzmann method for the advection-diffusion equation using dynamic circuits},
	author={D. Wawrzyniak and J. Winter and S. Schmidt and T. Indinger and C. F. Jan{\ss}en and U. Schramm and N. A. Adams},
	journal={Comput. Phys. Commun.},
	volume={317},
	pages={109856},
	year={2025},
	doi={10.1016/j.cpc.2025.109856}
}

@article{tiwari2025algorithmic,
	title={Algorithmic advances towards a realizable quantum lattice {B}oltzmann method},
	author={A. Tiwari and J. Iaconis and J. Jojo and S. Ray and M. Roetteler and C. Hill and J. Pathak},
	journal={arXiv:2504.10870},
	year={2025}
}

@article{xu2025improved,
	title={Improved quantum lattice {B}oltzmann method for advection-diffusion equations with a linear collision model},
	author={L. Xu and M. Li and L. Zhang and H. Sun and J. Yao},
	journal={Phys. Rev. E},
	volume={111},
	number={4},
	pages={045305},
	year={2025}
}

@article{kumar2025quantum,
	title={Quantum unitary matrix representation of the lattice {B}oltzmann model for low {R}eynolds fluid flow simulation},
	author={E. D. Kumar and S. H. Frankel},
	journal={AVS Quantum Sci.},
	volume={7},
	number={1},
	year={2025}
}

@article{nagel2026multiple,
 author={Aaron Nagel and Johannes L{\"o}we},
 title={Quantum lattice {Boltzmann} method for multiple time steps without reinitialization for linear advection-diffusion problems},
 journal={Comput. Phys. Commun.},
 volume={321},
 pages={110040},
 year={2026},
 doi={10.1016/j.cpc.2026.110040}
}

@article{novikau2025quantum,
 author={I. Novikau and I. Joseph},
 title={Quantum algorithm for the advection-diffusion equation and the {K}oopman-von {N}eumann approach to nonlinear dynamical systems},
 journal={Comput. Phys. Commun.},
 volume={309},
 pages={109498},
 year={2025},
 doi={10.1016/j.cpc.2025.109498}
}

@article{georgescu2025qlbm,
 author={C{\u{a}}lin A. Georgescu and Merel A. Schalkers and Matthias M{\"o}ller},
 title={{qlbm} -- {A} quantum lattice {Boltzmann} software framework},
 journal={Comput. Phys. Commun.},
 volume={315},
 pages={109699},
 year={2025},
 doi={10.1016/j.cpc.2025.109699}
}

@article{auyeung2024quantum,
 author={R. Au-Yeung and A. J. Williams and V. M. Kendon and S. J. Lind},
 title={Quantum algorithm for smoothed particle hydrodynamics},
 journal={Comput. Phys. Commun.},
 volume={294},
 pages={108909},
 year={2024},
 doi={10.1016/j.cpc.2023.108909}
}

@article{matyka2021memory,
 author={Maciej Matyka and Micha{\l} Dzikowski},
 title={Memory-efficient {Lattice Boltzmann Method} for low {R}eynolds number flows},
 journal={Comput. Phys. Commun.},
 volume={267},
 pages={108044},
 year={2021},
 doi={10.1016/j.cpc.2021.108044}
}

@article{cartalade2019multiple,
 author={Alain Cartalade and Amina Younsi and Marie-Christine N{\'e}el},
 title={{Multiple-Relaxation-Time Lattice Boltzmann} scheme for fractional advection--diffusion equation},
 journal={Comput. Phys. Commun.},
 volume={234},
 pages={40--54},
 year={2019},
 doi={10.1016/j.cpc.2018.08.005}
}

@article{goubault2020quantum,
 author={Timoth{\'e}e Goubault de Brugi{\`e}re and Marc Baboulin and Beno{\^i}t Valiron and Cyril Allouche},
 title={Quantum circuits synthesis using {H}ouseholder transformations},
 journal={Comput. Phys. Commun.},
 volume={248},
 pages={107001},
 year={2020},
 doi={10.1016/j.cpc.2019.107001}
}

@article{barenco1995elementary,
 author={Adriano Barenco and Charles H. Bennett and Richard Cleve and David P. DiVincenzo and Norman Margolus and Peter Shor and Tycho Sleator and John A. Smolin and Harald Weinfurter},
 title={Elementary gates for quantum computation},
 journal={Phys. Rev. A},
 volume={52},
 number={5},
 pages={3457--3467},
 year={1995},
 doi={10.1103/PhysRevA.52.3457}
}

@article{shende2006synthesis,
 author={Vivek V. Shende and Stephen S. Bullock and Igor L. Markov},
 title={Synthesis of quantum-logic circuits},
 journal={IEEE Trans. Comput.-Aided Des. Integr. Circuits Syst.},
 volume={25},
 number={6},
 pages={1000--1010},
 year={2006},
 doi={10.1109/TCAD.2005.855930}
}

@article{huang2014modified,
 author={Rongzong Huang and Huiying Wu},
 title={A modified multiple-relaxation-time lattice {Boltzmann} model for convection--diffusion equation},
 journal={J. Comput. Phys.},
 volume={274},
 pages={50--63},
 year={2014},
 doi={10.1016/j.jcp.2014.05.041}
}

@article{brassard2000amplitude,
 author={Gilles Brassard and Peter H{\o}yer and Michele Mosca and Alain Tapp},
 title={Quantum Amplitude Amplification and Estimation},
 journal={Contemp. Math.},
 volume={305},
 pages={53--74},
 year={2002},
 url={https://arxiv.org/abs/quant-ph/0005055}
}

@article{sterling1996stability,
 author={James D. Sterling and Shiyi Chen},
 title={Stability analysis of lattice {Boltzmann} methods},
 journal={J. Comput. Phys.},
 volume={123},
 number={1},
 pages={196--206},
 year={1996},
 doi={10.1006/jcph.1996.0016}
}

@article{chai2013lattice,
 author={Zhenhua Chai and T. S. Zhao},
 title={Lattice {Boltzmann} model for the convection-diffusion equation},
 journal={Phys. Rev. E},
 volume={87},
 number={6},
 pages={063309},
 year={2013},
 doi={10.1103/PhysRevE.87.063309}
}

@article{junk2009weighted,
 author={Michael Junk and Wen-An Yong},
 title={Weighted {$\mathbb{L}^2$}-stability of the lattice {Boltzmann} method},
 journal={SIAM J. Numer. Anal.},
 volume={47},
 number={3},
 pages={1651--1665},
 year={2009},
 doi={10.1137/060675216}
}

@article{rao2015numerical,
 author={Parthib R. Rao and Laura A. Schaefer},
 title={Numerical stability of explicit off-lattice {Boltzmann} schemes: {A} comparative study},
 journal={J. Comput. Phys.},
 volume={285},
 pages={251--264},
 year={2015},
 doi={10.1016/j.jcp.2015.01.017}
}

@article{bellotti2024modified,
 author={Thomas Bellotti},
 title={Initialisation from lattice {Boltzmann} to multi-step finite difference methods: {M}odified equations and discrete observability},
 journal={J. Comput. Phys.},
 volume={504},
 pages={112871},
 year={2024},
 doi={10.1016/j.jcp.2024.112871}
}

@article{berry2014high,
 author={Dominic W. Berry},
 title={High-order quantum algorithm for solving linear differential equations},
 journal={J. Phys. A: Math. Theor.},
 volume={47},
 number={10},
 pages={105301},
 year={2014},
 doi={10.1088/1751-8113/47/10/105301}
}

@article{jin2022quantum,
 author={Shi Jin and Nana Liu and Yue Yu},
 title={Time complexity analysis of quantum difference methods for linear high dimensional and multiscale partial differential equations},
 journal={J. Comput. Phys.},
 volume={471},
 pages={111641},
 year={2022},
 doi={10.1016/j.jcp.2022.111641}
}

@article{schalkers2024importance,
 author={Merel A. Schalkers and Matthias M{\"o}ller},
 title={On the importance of data encoding in quantum {Boltzmann} methods},
 journal={Quantum Inf. Process.},
 volume={23},
 number={1},
 pages={20},
 year={2024},
 doi={10.1007/s11128-023-04216-y}
}

@article{georgescu2025fully,
 author={C{\u{a}}lin A. Georgescu and Merel A. Schalkers and Matthias M{\"o}ller},
 title={Fully quantum lattice gas automata building blocks for computational basis state encodings},
 journal={J. Comput. Phys.},
 volume={544},
 pages={114595},
 year={2026},
 doi={10.1016/j.jcp.2025.114595}
}

@article{jin2026schrodingerization,
 author={Shi Jin and Nana Liu and Yue Yu},
 title={Schr{\"o}dingerization based quantum algorithms for the time-fractional heat equation},
 journal={J. Comput. Phys.},
 volume={560},
 pages={114943},
 year={2026},
 doi={10.1016/j.jcp.2026.114943}
}

@article{novikau2026explicit,
 author={I. Novikau and I. Joseph},
 title={An efficient explicit implementation of a near-optimal quantum algorithm for simulating linear dissipative differential equations},
 journal={J. Comput. Phys.},
 volume={564},
 pages={115132},
 year={2026},
 doi={10.1016/j.jcp.2026.115132}
}

@article{yao2026representation,
  author={Hua-Dong Yao and Sauro Succi},
  title={Representation-split {Carleman} lattice {Boltzmann} with near-unit step success probability on quantum computers},
  journal={arXiv:2610.20620},
  year={2026},
  url={https://arxiv.org/abs/2610.20620}
}

@article{over2025direct,
  author={Paul Over and Sergio Bengoechea and Peter Brearley and Sylvain Laizet and Thomas Rung},
  title={Quantum algorithm for the advection-diffusion equation by direct block encoding of the time-marching operator},
  journal={Phys. Rev. A},
  volume={112},
  pages={L010401},
  year={2025},
  doi={10.1103/d8hb-fv93}
}

\end{document}